\documentclass[11pt]{article}

\usepackage{algpseudocode, algorithm, amsmath, amssymb, amsthm, bbm, booktabs, comment, float, graphicx, mathrsfs, rotating, url}    
\usepackage[a4paper]{geometry}
\usepackage[round]{natbib}
\usepackage[table]{xcolor}

\usepackage[hidelinks]{hyperref}

\theoremstyle{plain}
\newtheorem{thm}{Theorem}[section]
\newtheorem{cor}[thm]{Corollary} 
 
\newtheorem{prop}[thm]{Proposition} 

\theoremstyle{definition} 
 
\newtheorem{assump}{Assumption}[section]
 
\newtheorem{example}{Example}[section]

\theoremstyle{remark} 
\newtheorem{rem}[thm]{Remark}

\newcommand{\cdf}{{\rm cdf}}

\newcommand{\N}{\mathbb{N}} 
\newcommand{\R}{\mathbb{R}} 

\newcommand{\cO}{\mathcal{O}} 
\newcommand{\diff}{\,\mathrm{d}} 
\newcommand{\eqand}{\quad \textrm{and } \quad}

\newcommand{\Mid}{\hspace{1mm} \big| \hspace{1mm}}
\newcommand{\MMid}{\hspace{1mm} \Big| \hspace{1mm}}
\newcommand{\one}{\mathbbm{1}} 
\newcommand{\supp}{{\rm supp}} 

\newcommand{\E}{\mathbb{E}} 
\newcommand{\prob}{\mathbb{P}} 
\newcommand{\PR}{\mathcal{P}(\R)} 

\newcommand{\A}{\mathcal{A}} 
\newcommand{\BR}{\mathcal{B}(\R)} 
\newcommand{\BPR}{\mathcal{B}(\PR)} 
\newcommand{\BU}{\mathcal{B}(0,1)} 
\newcommand{\F}{\mathcal{F}} 
\newcommand{\LL}{\mathscr{L}} 

\renewcommand{\phi}{\varphi} 
\newcommand{\s}{\sigma} 

\newcommand{\mg}{{\rm mg}} 
\newcommand{\st}{\leq_\textrm{st}} 
\newcommand{\XX}{\mathcal{X}} 

\newcommand{\crps}{{\rm crps}} 
\newcommand{\Sbar}{\overline{\rm CRPS}} 

\newcommand{\bs}{\mathrm{s}_\mathrm{B}} 
\newcommand{\BS}{{\rm BS}} 
\newcommand{\BSbar}{\overline{{\rm BS}}} 

\newcommand{\qs}{\mathrm{qs}_\alpha} 
\newcommand{\QS}{{\rm QS}} 
\newcommand{\QSbar}{\overline{{\rm QS}}} 

\newcommand{\DSC}{{\rm DSC}} 
\newcommand{\MCB}{{\rm MCB}} 
\newcommand{\MS}{{\rm MS}} 
\newcommand{\UNC}{{\rm UNC}} 

\newcommand{\DSCbar}{\overline{{\rm DSC}}} 
\newcommand{\MCBbar}{\overline{{\rm MCB}}} 
\newcommand{\UNCbar}{{\overline{\rm UNC}}} 

\newcommand{\CT}{{\rm CT}} 
\newcommand{\HB}{{\rm HB}} 
\newcommand{\HBO}{{\rm HBo}} 
\newcommand{\ISO}{{\rm ISO}} 
\newcommand{\ab}{{(a,b)}} 

\newcommand{\cF}{{\acute{F}}} 
\newcommand{\cFinv}{{\grave{F}}} 

\newcommand\myrot[1]{\mathrel{\rotatebox[origin=c]{#1}{$\Rightarrow$}}}
  \newcommand\SWarrow{\myrot{-135}} \newcommand\SEarrow{\myrot{-45}}

\newcommand\myrotation[1]{\mathrel{\rotatebox[origin=c]{#1}{$\hookrightarrow$}}}
   \newcommand\SEArrow{\myrotation{-45}}

\newcommand\myrotations[1]{\mathrel{\rotatebox[origin=c]{#1}{$\hookleftarrow$}}}
\newcommand\NEArrows{\myrotations{45}}   

\begin{document}
	
\title{Decompositions of the mean \\ continuous ranked probability score}

\author{Sebastian Arnold\thanks{Both authors contributed equally to this work.} \thanks{University of Bern, \texttt{sebastian.arnold@unibe.ch}}, Eva-Maria Walz\footnotemark[1] \thanks{Karlsruhe Institute of Technology (KIT) and Heidelberg Institute for Theoretical Studies (HITS), \texttt{eva-maria.walz@kit.edu}}, Johanna Ziegel\thanks{University of Bern, \texttt{johanna.ziegel@unibe.ch}}, Tilmann Gneiting\thanks{Heidelberg Institute for Theoretical Studies (HITS) and Karlsruhe Institute of Technology (KIT), \texttt{tilmann.gneiting@h-its.org}}}
	
\maketitle
	
\begin{abstract}
The continuous ranked probability score (\crps) is the most commonly used scoring rule in the evaluation of probabilistic forecasts for real-valued outcomes.  To assess and rank forecasting methods, researchers compute the mean \crps\ over given sets of forecast situations, based on the respective predictive distributions and outcomes.  We propose a new, isotonicity-based decomposition of the mean \crps\ into interpretable components that quantify miscalibration (MSC), discrimination ability (DSC), and uncertainty (UNC), respectively.  In a detailed theoretical analysis, we compare the new approach to empirical decompositions proposed earlier, generalize to population versions, analyse their properties and relationships, and relate to a hierarchy of notions of calibration.  The isotonicity-based decomposition guarantees the nonnegativity of the components and quantifies calibration in a sense that is stronger than for other types of decompositions, subject to the nondegeneracy of empirical decompositions.  We illustrate the usage of the isotonicity-based decomposition in case studies from weather prediction and machine learning.  
\end{abstract}	

\section{Introduction}  \label{sec:introduction}

Probabilistic predictions are forecasts in the form of predictive probability distributions, which ought to be as sharp as possible subject to calibration \citep{Gneiting2007b}.  Informally, predictive distributions are calibrated if they provide a statistically coherent explanation of the outcomes.  Sharpness, on the other hand, quantifies how well one can discriminate different scenarios for future events according to the forecast and is a property of the forecast only.  For the comparative evaluation of probabilistic forecasts, proper scoring rules should be employed \citep{Gneiting2007a}.  A proper scoring rule assigns a numerical score to a probabilistic forecast with corresponding observed realization, and addresses calibration and sharpness simultaneously.  If we compare two competing forecasts according to their scores, it is natural to ask in which aspect one forecast is superior to the other.  This motivates the decomposition of average realized scores into more interpretable terms measuring calibration, discrimination ability, and uncertainty, respectively. 

Historically, the first score decomposition was introduced by \citet{Murpy_1973}, who proposed a decomposition of the mean Brier score (BS).  For a sequence of forecast--observation pairs $(p_1,y_1), \dots, (p_n,y_n)$, consisting of predictive probabilities $p_i \in [0,1]$ and corresponding binary outcomes $y_i \in \{ 0, 1 \}$, the empirical average Brier score 
\begin{align*}
\BSbar = \frac{1}{n} \sum_{1=1}^n \left( p_i - y_i \right)^2
\end{align*}
quantifies the overall performance of the assessed forecasts based on the actual observations.  \citet{Murpy_1973} motivates a decomposition of $\BSbar$ into interpretable components: a term measuring miscalibration (MCB) or reliability, a term measuring discrimination ability (DSC) or resolution, and a term quantifying the overall uncertainty (UNC) of the outcome.  Originally derived as a vector partition by \citet{Murpy_1973}, \citet{Siegert} gives a persuasive interpretation of the Murphy decomposition: For $k = 1, \dots, n$, consider the conditional event probability $q_k$, i.e., the proportion of realized binary events ($y_i = 1$) in the cases where the forecast was $p_k$.  Denote by $\BSbar_c$ the empirical Brier score of the calibrated forecasts $q_1, \dots, q_k$, and by $\BSbar_r$ the empirical Brier score with respect to the static reference forecast $r = (1/n) \sum_{i=1}^n y_i$, namely, 
\begin{equation}  \label{eq:Murphy_components}
\BSbar_c = \frac{1}{n} \sum_{1=1}^n \left( q_i - y_i \right)^2 \eqand \BSbar_r = \frac{1}{n} \sum_{1=1}^n \left( r - y_i \right)^2.  
\end{equation}
\citet{Siegert} shows that the Murphy decomposition reads as 
\begin{equation}  \label{eq:Murphy_decomposition}
\BSbar = \underbrace{\big( \BSbar - \BSbar_c \big)}_{\MCBbar} - \underbrace{\big( \BSbar_r - \BSbar_c \big)}_{\DSCbar} + \underbrace{\BSbar_r}_{\UNCbar}.
\end{equation}
The three terms of this exact decomposition reveal deeper insight into the performance of the assessed forecasts: The predictive probabilities are calibrated if they are close to their conditional event probabilities, and hence, low values of $\MCBbar$ indicate a good performance in terms of calibration.  A perfectly calibrated forecast sequence can be constructed by issuing the marginal probability $r$ over all instances.  Even though perfectly calibrated, such a sequence would not be informative, since the same predictive probability is issued throughout.  For such a sequence, we would obtain $\DSCbar = 0$, which has a negative effect on the score, whereas larger values of $\DSCbar$ are obtained if the calibrated forecasts can discriminate different scenarios better than the reference forecast.  Finally, the $\UNCbar$ component informs about the inherent difficulty of the prediction problem and is independent of the forecasts. 

The rationale behind the decomposition in \eqref{eq:Murphy_decomposition} can be summarized as the following recipe: Having available a calibration method that transforms the original forecasts $p_1, \dots, p_n$ into calibrated forecasts $q_1, \dots, q_n$, one can measure miscalibration as the difference in the mean score of the original forecasts to the calibrated ones, resulting in the $\MCBbar$ term.  The CORP (Consistent, Optimally binned, Reproducible, and PAV algorithm based) score decomposition suggested by \cite{CORP} uses this general recipe, where the calibrated forecasts $q_1, \dots, q_n$ are computed by applying nonparametric isotonic regression on the vector $(y_1, \dots, y_n)$ with respect to the order induced by $(p_1, \dots, p_n)$.  The authors argue  that ``the assumption of nondecreasing CEPs is natural, as decreasing estimates are counterintuitive, routinely being dismissed as artifacts by practitioners'' \cite[p.~4]{CORP}.  If we consider, e.g., the conditional event probability over all events where we predicted a positive outcome with probability $0.5$, then we should expect this value to be smaller than the conditional event probability over all events where we predicted a positive outcome with probability $0.6$.  As noted by \citet{Bentzien_Friederichs_2014}, \citet{Siegert}, \citet{LH2021}, and \citet{GLS2023}, and discussed in detail by \citet{Tilmann_Johannes_Calibration}, the recipe extends to scores other than the Brier score and general types of statistical functionals. 

In this paper, we focus on the continuous ranked probability score (\crps; \citealp{Matheson_Winkler_1976}).  The \crps\ is one of the most prominent scoring rules for the evaluation of probabilistic forecasts for real-valued outcomes and is popular across application areas and methodological communities; see, e.g., \citet{Gneiting2005b}, \citet{Hothorn2014}, \citet{Pappenberger2015}, \citet{Rasp2018}, and \citet{Gasthaus2019}.  The \crps\ is defined in terms of any cumulative distribution function (cdf) $F$ on $\R$ and $y \in \R$, and given by 
\begin{align*}
\crps(F,y) = \int_{\R} \big( F(z) - \one \{ y \leq z \} \big)^2 \diff z.
\end{align*}
For a sequence of forecast--observation pairs $(F_1,y_1), \dots, (F_n,y_n)$, comprising a predictive distribution $F_i$ and a corresponding real-valued outcome $y_i$, the mean \crps, 
\begin{equation}  \label{eq:average_CRPS}
\Sbar = \frac{1}{n} \sum_{i=1}^n \crps(F_i,y_i)
\end{equation}
serves to quantify the overall performance of the forecasts.  Possible decompositions of the mean score at \eqref{eq:average_CRPS} have been discussed in the literature, with the most prominent approaches being introduced by \citet{Hersbach_2000} and \citet{Candille_2005}.  These methods offer promising solutions but come with severe limitations.  In a nutshell, the Hersbach decomposition lacks a theoretical background and the desirable property that the components of the decomposition are nonnegative, whereas the decomposition of \citet{Candille_2005} is not practically feasible, as acknowledged by the authors.  Another approach for decomposing the mean \crps\ is by exploiting its representation as an integral over Brier scores, compare \eqref{eq:CRPS_BS}, and then integrating existing decompositions of $\BSbar$.  Similarly, the \crps\ can be expressed as an integral over quantile scores, see \eqref{eq:CRPS_QS}, and existing decompositions for quantile scores can be leveraged to decompose the mean score at \eqref{eq:average_CRPS}.  However, these approaches have the drawback that miscalibration and discrimination ability are not measured with respect to the full probabilistic forecasts but only with respect to individual threshold or quantile levels.  

In this article, we propose a new decomposition of the mean \crps\ based on Isotonic Distributional Regression \citep[IDR;][]{IDR}.  In the case of binary outcomes, \citet{CORP} argue that isotonicity between the predictive probabilities and the calibrated forecasts is a natural constraint, since violations of isotonicity lead to poor predictive performance.  This argument generalizes to the real-valued setting, since it is natural to assume that the conditional law of the outcome, given the forecast, should tend to be small (large) if the predictive distribution is small (large), where notions of small and large are understood with respect to the usual stochastic order.  IDR is a nonparametric distributional regression technique that honors the shape constraint of isotonicity between covariates and responses.  Applying IDR to the data $(F_1,y_1), \dots, (F_n,y_n)$ yields calibrated forecasts, whereas the marginal distribution of the outcomes $y_1, \ldots, y_n$ serves as static reference forecast.  The general recipe from \eqref{eq:Murphy_components} and \eqref{eq:Murphy_decomposition} then yields mean scores for the calibrated forecast and the reference forecast, respectively, and a corresponding exact decomposition, 
\begin{align*}
\Sbar = \MCBbar_\ISO - \DSCbar_\ISO + \UNCbar_0,
\end{align*}
of the mean \crps\ at \eqref{eq:average_CRPS}, to which we refer as the isotonicity-based decomposition.  The isotonicity-based approach guarantees the nonnegativity of the three components, and the miscalibration term admits a persuasive interpretation in terms of calibration. 

While auto-calibration serves as the universal notion of calibration for binary events \citep[Theorem 2.11]{Gneiting_Ranjan_2013}, for real-valued random outcomes, numerous different notions of calibration are found in the literature \citep{Dawid1984, Diebold1998, Strahl_Ziegel_2017, Arnold_Henzi_Ziegel}, as reviewed by \citet{Tilmann_Johannes_Calibration}.  The strongest notion is auto-calibration and, ideally, one would like to measure miscalibration as deviation from auto-calibration, as targeted by the decomposition of \citet{Candille_2005}.  However, the Candille--Talagrand approach yields degenerate empirical decompositions.  Therefore, we quantify miscalibration as the deviation from isotonic calibration, as introduced by \cite{ICL} in a study of the population version of IDR.  Isotonic calibration is closer to auto-calibration than the notions of calibration targeted by the Hersbach decomposition, or by the aforementioned decompositions based on Brier or quantile scores. 

The remainder of the paper is organized as follows.  Section \ref{sec:empirical} reviews the previously proposed  decompositions and their properties.  In Section \ref{sec:ISO}, we develop the empirical version of the new isotonicity-based decomposition, followed by a thorough study of the population versions of the various types of decomposition and their properties in Section \ref{sec:population}, with particular emphasis on calibration.  In Section \ref{sec:case_study}, we apply the proposed isotonicity-based decomposition in case studies from meteorology and machine learning.  The main part of the paper closes with a discussion in Section \ref{sec:discussion}.  Proofs, technical comments, and a series of detailed analytic examples in population settings are available in Appendices \ref{app:BS_QS} through \ref{app:proofs}.

\section{Previously proposed empirical decompositions}  \label{sec:empirical}

\subsection{Preliminaries}  \label{sec:preliminaries}

Throughout the article, we denote by $\PR$ the class of all probability distributions on $\R$ with finite first moment.  We treat its elements interchangeably as probability measures or cumulative distribution functions ({\cdf}s).

Single-valued forecasts for functionals of an unknown quantity should be compared using consistent scoring functions \citep{Gneiting2011a}.  For example, the \textit{quadratic score} ${\rm s}(x,y) = (x-y)^2$, and the piecewise linear \textit{quantile score}
\begin{align}  \label{eq:QS}
\qs(x,y) = (\one \{ y \leq x \} - \alpha) \, (x - y),
\end{align}
where $x, y \in \R$, are consistent scoring functions for the mean functional, and for the quantile at level $\alpha \in (0,1)$, respectively.  In other words, $\int (x-y)^2 \diff F(y)$ is minimal when $x$ is the mean of $F \in \PR$, and $\int \qs(x,y) \diff F(y)$ is minimal when $x$ is a quantile of $F$ at level $\alpha \in (0,1)$.  

Probabilistic forecasts specify a probability measure over all possible values of the outcome, and predictive performance ought to be be compared and evaluated using proper scoring rules \citep{Gneiting2007a}.  A popular proper scoring rule for probability forecasts of a binary outcome is the \textit{Brier score} 
\begin{align}  \label{eq:BS}
\bs(p,y) = (p-y)^2,
\end{align}
where $p \in [0,1]$ and $1 - p$ are the predicted probabilities of the outcomes $y = 1$ and $y = 0$, respectively.  A key example of a proper scoring rule for predictive distributions over $\R$ is the \textit{continuous ranked probability score} (crps), defined for all $F \in \PR$ and $y \in \R$, and given equivalently by
\begin{align}
\crps(F,y) & = \int \bs(F(z), \one \{ y \leq z\}) \diff z \label{eq:CRPS_BS} \\
           & = \int_0^1 \qs(F^{-1}(\alpha), y) \diff \alpha \label{eq:CRPS_QS},
\end{align}
where $\bs$ and $\qs$ are defined at \eqref{eq:BS} and \eqref{eq:QS}, respectively, and where $F^{-1}$ denotes the quantile function defined as $F^{-1}(\alpha) = \inf\{ z \in \R \mid F(z) \geq \alpha \}$ for $\alpha \in (0,1)$.  The representation at \eqref{eq:CRPS_QS} is due to \citet{Laio_Tamea_2007}.

We consider a collection
\begin{equation}  \label{eq:data} 
(F_1, y_1), \dots, (F_n, y_n)
\end{equation}
of tuples that comprise a forecast $F_i \in \PR$ in the form of a \cdf\ and the respective outcome $y_i \in \R$, where $i = 1, \dots, n$.  Our aim is to decompose the empirical mean score, 
\begin{equation}  \label{eq:meanCRPS}
\Sbar = \frac{1}{n} \sum_{i=1}^n \crps(F_i,y_i),
\end{equation}
of the forecast--observation pairs at \eqref{eq:data} into three distinct components, namely, miscalibration ($\MCBbar$), discrimination ($\DSCbar$), and uncertainty ($\UNCbar$).  The following desirable properties are relevant.  
\begin{itemize}  
\item[($E_1$)] The decomposition is exact, i.e.,
\begin{align*}
\Sbar = \MCBbar - \DSCbar + \UNCbar.
\end{align*}
\item[($E_2$)] The components $\MCBbar$, $\DSCbar$, and $\UNCbar$ are nonnegative.
\item[($E_3$)] The decomposition is not degenerate. Here, a decomposition is \textit{degenerate} if $\MCBbar = 0$ whenever $F_1, \dots, F_n$ are pairwise distinct.  
\item[($E_4$)] The $\DSCbar$ component vanishes if $F_1 = \dots = F_n$. 
\item[($E_5$)] The $\UNCbar$ component can be expressed in terms of the outcomes $y_1, \dots, y_n$ only.
\end{itemize}
These conditions do not depend on the use of any specific scoring rule; they are desirable for decompositions of mean scores in general.

An exact decomposition ($E_1$) is desirable, since it allows us to fully decompose the mean score.  A degenerate decomposition is undesirable, as in typical practice, such as in the case studies in Section \ref{sec:case_study}, the issued forecast distributions are pairwise distinct, and then the method is useless.  A static forecast, i.e., $F_1 = \dots = F_n$, has no discrimination ability, hence ($E_4$) is desirable.  Requirement ($E_5$) is natural since intrinsic uncertainty does not depend on the activities of forecasters.

Finally, we argue that there ought to be a population version of the decomposition that applies to any admissible joint distribution $\prob$ of tuples $(F,Y)$.  Furthermore, the population version ought to reduce to the empirical version if $\prob$ is the empirical measure for the data at \eqref{eq:data}.  We study decompositions at the population level in Section \ref{sec:population}.

\subsection{Candille--Talagrand decomposition}  \label{sec:CT} 

\citet{Candille_2005} naturally extend the idea of the Murphy decomposition at \eqref{eq:Murphy_decomposition}.  To describe their approach, let $\delta_y$ denote the Dirac or point measure in $y \in \R$, and let the marginal law $\hat{F}_\mg = \frac{1}{n} \sum_{i=1}^n \delta_{y_i}$ denote the empirical distribution of the outcomes $y_1, \dots, y_n$ in \eqref{eq:data}.  Let $\hat{F}_i$ be the auto-calibrated version of the forecast $F_i$ in \eqref{eq:data}, i.e., let $\hat{F}_i$ be the normalized version of $\sum_{j=1}^n \one \{ F_j = F_i \} \, \delta_{y_j}$ for $i = 1, \dots, n$.  Then
\begin{equation}  \label{eq:Sbar_mg_ac}
\Sbar_\mg = \frac{1}{n} \sum_{i=1}^n \crps(\hat{F}_\mg,y_i)
\eqand
\Sbar_\textrm{ac} = \frac{1}{n} \sum_{i=1}^n \crps(\hat{F}_i,y_i)
\end{equation}
are the mean score of the marginal forecast and the auto-calibrated forecast, respectively.  \citet{Candille_2005} define uncertainty, miscalibration, and discrimination components as 
\begin{equation}  \label{eq:UNC}
\UNCbar_0 = \Sbar_\mg, 
\end{equation}
\begin{equation}  \label{eq:CTcomponents}
\MCBbar_\CT = \Sbar - \Sbar_\textrm{ac}, \qquad \DSCbar_\CT = \Sbar_\mg - \Sbar_\textrm{ac}, 
\end{equation}
respectively, to yield the \textit{Candille--Talagrand} (CT) \textit{decomposition}  
\begin{equation}  \label{eq:CTdecomposition}
\Sbar = \MCBbar_\CT - \DSCbar_\CT + \UNCbar_0. 
\end{equation}
The Candille--Talagrand decomposition tackles the core idea of auto-calibration and satisfies properties ($E_1$), ($E_2$), ($E_4$), and ($E_5$), but fails to satisfy the nondegeneracy condition ($E_3$), which prohibits its practical use.   

To avoid a degenerate decomposition, one might  partition the forecasts into equivalence classes of {\cdf}s that are considered identical when calibrating \citep[p.~2147]{Candille_2005}.  However, the choice of such a partition is challenging and the decomposition depends on its effects, akin to the effects of binning on the classical reliability diagram for probability forecasts of a binary event as described by \citet{CORP} and references therein.

\subsection{Brier score based decomposition}  \label{sec:BS} 

The Brier score based representation of individual \crps\ values at \eqref{eq:CRPS_BS} implies that 
\begin{equation}  \label{eq:Sbar_BS}
\Sbar = \frac{1}{n} \sum_{i=1}^n \crps(F_i,y_i) = \int_{- \infty}^\infty \! \BSbar_z \diff z, 
\end{equation}
where
\begin{align*}
\BSbar_{z} = \frac{1}{n} \sum_{i=1}^n \bs( F_i(z), \one \{ y_i \leq z \}). 
\end{align*}
In this light, a natural way of decomposing $\Sbar$ lies in integrating a given decomposition of the mean Brier score, as proposed and implemented by \citet{Ferro_Fricker_2012}, \citet{Todter_Ahrens_2012}, and \citet{Lauret_et_al_2019}, among other authors.

Specifically, suppose that, for each $z \in \R$, there is a decomposition $\BSbar_{z} = \MCBbar_{\BS,z} - \DSCbar_{\BS,z} + \UNCbar_{\BS,z}$ of the mean Brier score.  Then we can define 
\begin{equation}  \label{eq:BScomponents}
\MCBbar_\BS = \int_{- \infty}^\infty \! \MCBbar_{\BS,z} \diff z, \;
\DSCbar_\BS = \int_{- \infty}^\infty \! \DSCbar_{\BS,z} \diff z, \;
\UNCbar_\BS = \int_{- \infty}^\infty \! \UNCbar_{\BS,z} \diff z. 
\end{equation}
The CORP approach of \citet{CORP} yields a compelling decomposition of the mean Brier score, which does neither require tuning, nor binning of the assessed predictive probabilities, and enforces a natural shape constraint of isotonicity between the predictive probabilities and the calibrated forecasts.  Throughout this article, we decompose the mean Brier score by the CORP approach and refer to the induced decomposition, namely,  
\begin{equation}  \label{eq:BSempirical}
\Sbar = \MCBbar_\BS - \DSCbar_\BS + \UNCbar_\BS, 
\end{equation}
as the \textit{Brier score based} (BS) decomposition of $\Sbar$.  Details of this approach are reviewed in Appendix \ref{app:BS}, where we prove the following result.

\begin{prop}  \label{prop:BS}
For the Brier score based decomposition at \eqref{eq:BSempirical} it holds that\/ $\UNCbar_\BS = \UNCbar_0$, and the decomposition satisies properties ($E_1$), ($E_2$), ($E_3$), ($E_4$), and ($E_5$).
\end{prop}

Despite these favorable properties, the Brier score based decomposition is subject to shortcomings and inconsistencies, due to the isolated treatment of probability forecasts at fixed thresholds.  For discussion, we refer the reader to Section \ref{sec:numerical_example} and Appendix \ref{app:BS_QS}. 

\subsection{Quantile score based decomposition}  \label{sec:QS}

In view of the quantile score representation of the \crps\ at \eqref{eq:CRPS_QS}, a natural approach to decomposing the mean score $\Sbar$ leverages decompositions of the mean quantile score at \eqref{eq:QS}.  Specifically, the quantile score representation implies that
\begin{align*}
\Sbar = \frac{1}{n}\sum_{i=1}^n \crps(F_i,y_i) = \int_{-\infty}^\infty \QSbar_\alpha \diff\alpha,
\end{align*}
where
\begin{align*}
\QSbar_\alpha = \frac{1}{n}\sum_{i=1}^n \qs(F_i^{-1}(\alpha),y_i).
\end{align*}
Suppose that for each $\alpha \in (0,1)$, there is a decomposition $\QSbar_\alpha = \MCBbar_{\QS,\alpha} - \DSCbar_{\QS,\alpha} + \UNCbar_{\QS,\alpha}$ of the mean quantile score, and define $\MCBbar_\QS$ as the integral of $\MCBbar_{\QS,\alpha}$ over $\alpha \in (0,1)$, and similarly for the discrimination and uncertainty components.  The CORP score decomposition of \citet{CORP} and its core idea of isotonicity as a shape constraint between issued and calibrated forecasts extend naturally to quantiles, as discussed by \citet[Section 3.3]{Tilmann_Johannes_Calibration} and \citet[Section 3.3]{Gneiting_2023}.  Throughout the article, we decompose the mean quantile score by the CORP approach and refer to the resulting decomposition, namely,
\begin{equation}  \label{eq:QSempirical}
\Sbar = \MCBbar_\QS - \DSCbar_\QS + \UNCbar_\QS, 
\end{equation}
as the \textit{quantile score based} (QS) decomposition of $\Sbar$.  For details, we refer the reader to Appendix \ref{app:QS} where we prove the following result.  

\begin{prop}  \label{prop:QS}
For the quantile score based decomposition at \eqref{eq:QSempirical} it holds that\/ $\UNCbar_\QS = \UNCbar_0$, and the decomposition satisfies properties ($E_1$), ($E_2$), ($E_3$), ($E_4$), and ($E_5$).
\end{prop}  

The quantile score based decomposition is subject to shortcomings in analogy to the issues with the Brier score based approach, due to the reliance on quantile forecasts at fixed levels; for further discussion see Section \ref{sec:numerical_example} and Appendix \ref{app:BS_QS}. 

\subsection{Hersbach decomposition }  \label{sec:HB}

The decomposition of \citet{Hersbach_2000} applies specifically to ensemble forecasts and operates under the implicit assumption of a continuous outcome.  For the data at \eqref{eq:data}, Hersbach's assumptions imply, without loss of generality, that for $i = 1, \dots, n$ the forecast $F_i$ is the empirical \cdf\ of a fixed number $m$ of values $x_1^i \le \dots \le x_m^i$, with the outcome $y_i \not \in \{ x_1^i, \dots, x_m^i \}$ being distinct from these values.  However, with a view towards a generalization of the Hersbach decomposition, we allow for any real-valued outcome $y_i$.  Figure \ref{fig:HB} in Appendix \ref{app:HB} illustrates in detail how the case $y_i \in \{ x_1^i, \dots, x_m^i \}$ should be handled in the Hersbach decomposition.

In line with the other types of decomposition, \citet{Hersbach_2000} defines the uncertainty component as $\UNCbar_0$ at \eqref{eq:UNC}.  The miscalibration component, which \citet{Hersbach_2000} refers to as reliability, is
\begin{align*}
\MCBbar_\HBO = \sum_{\ell=0}^m \bar{g}_\ell \left( p_\ell - \bar{o}_\ell \right)^2,
\end{align*}
where $p_\ell = \ell/m$ for $\ell = 0, \dots, m$, and $\bar{g}_\ell$ is the average width of bin $i$, i.e.,
\begin{equation}  \label{eq:g_ell}
\bar{g}_\ell = \frac{1}{n} \sum_{i=1}^n (x_{\ell+1}^i - x_\ell^i)
\end{equation}
for $\ell = 1, \dots, m-1$.  The term $\bar{o}_\ell$ approximates the average frequency of an outcome below the midpoint of bin $\ell$; specifically, 
\begin{align*}
\bar{o}_\ell = \bar{f}_\ell - \bar{m}_\ell,
\end{align*}
where 
\begin{align}  \label{eq:f_ell}
\bar{f}_\ell = \frac{1}{n\bar{g}_\ell} \sum_{i=1}^n \one \{ F_i(y_i) \le p_\ell\} \, (x_{\ell+1}^i - x_\ell^i)
\;\;\text{and}\;\;
\bar{m}_\ell = \frac{1}{n\bar{g}_\ell} \sum_{i=1}^n \one \{ x_\ell^i < y_i < x_{\ell+1}^i \} \, (y_i - x_\ell^i)
\end{align}
for $\ell = 1, \dots, m - 1$.  For any $\ell$ with $x_{\ell}^i < x_{\ell+1}^i$ it holds that $F_i(y_i) \le p_l$ if, and only if, $y_i < x_{\ell+1}^i$.  To complete the specification, we let $\bar{o}_0 = (1/n) \sum_{i=1}^n \one \{ y_i < x_1^i \}$ and $\bar{o}_m = (1/n) \sum_{i=1}^n \one \{ x_m^i < y_i \}$, and if these quantities are nonzero then we let $\bar{g}_0 = (1/(n \bar{o}_0)) \sum_{i=1}^n \one \{ y_i < x_1^i \} \, (x_1^i - y_i)$ and $\bar{g}_m = (1/(n \bar{o}_m)) \sum_{i=1}^n \one \{ x_m^i < y_i \} \, (y_i - x_m^i)$.  The miscalibration component thus measures deviations from uniformity for the rank histogram \citep{Hamill_2001, Gneiting2007b}.

\citet{Hersbach_2000} defines the resolution (in our terminology, the discrimination) component $\DSCbar_\HBO = \MCBbar_\HBO + \UNCbar_0 - \Sbar$ as the remainder, to complete the \textit{original Hersbach} (\HBO) \textit{decomposition} 
\begin{equation}  \label{eq:Hersbach} 
\Sbar = \MCBbar_\HBO - \DSCbar_\HBO + \UNCbar_0.  
\end{equation}
Towards a generalization, we introduce a slightly modified miscalibration component, 
\begin{equation}  \label{eq:MCB_HB}
\MCBbar_{\HB} = \sum_{\ell=1}^{m-1} \bar{g}_\ell \left( p_\ell - \bar{f}_\ell \right)^2, 
\end{equation}
and a respectively modified discrimination component, $\DSCbar_{\HB} = \MCBbar_{\HB} + \UNCbar_0 - \Sbar$, to yield the \textit{modified Hersbach}, or simply \textit{Hersbach} (\HB) \textit{decomposition},  
\begin{equation}  \label{eq:Hersbach_mod} 
\Sbar = \MCBbar_\HB - \DSCbar_\HB + \UNCbar_0.  
\end{equation}
The interpretation of the miscalibration component remains unchanged, as $\MCBbar_\HB$ and $\MCBbar_\HBO$ differ only slightly, with $\bar{f}_\ell$ in \eqref{eq:MCB_HB} being the approximate frequency of an outcome below the right endpoint of bin $\ell$.  For a more detailed comparison and the proof of the following result, we refer the reader to Appendix \ref{app:HB}.

\begin{prop}  \label{prop:HB_decomp}
The original and modified Hersbach decompositions at \eqref{eq:Hersbach} and \eqref{eq:Hersbach_mod}, respectively, satisfy properties ($E_1$), ($E_3$), and ($E_5$), while properties ($E_2$) and ($E_4$) fail to hold.
\end{prop}

As discussed thus far, the Hersbach decomposition requires that the forecasts assume the form of an ensemble.   Further shortcomings have been discussed in the literature \citep{Siegert}; in particular, it has been noted that the discrimination component $\DSCbar_\HBO$ is defined ``somewhat artificially'' \cite[p.~565]{Hersbach_2000} and that it can be negative, thus violating $(E_2)$.  The original Hersbach decomposition has been extended by Lalaurette so that it applies to forecasts with strictly increasing {\cdf}s \citep[Appendix A]{Candille_2005}.  We discuss and generalize Lalaurette's extension in Section \ref{sec:decomposition}, and our analysis demonstrates that the extensions can more naturally be interpreted as extensions of the modified Hersbach decomposition.  In Appendix \ref{app:proofs1} we describe empirical versions that apply in the general case of forecast distributions with finite support, and to mixed discrete-continuous distributions for nonnegative quantities, respectively. 

\subsection{Numerical example and discussion}  \label{sec:numerical_example}

For illustration, we consider forecasts from the case studies in Section \ref{sec:case_study}.  The decompositions from Sections \ref{sec:CT} through \ref{sec:HB} all use the uncertainty component $\UNCbar_0$ at \eqref{eq:UNC}, and they specify the discrimination component as 
\begin{align*}
\DSCbar_\bullet = \Sbar - \MCBbar_\bullet - \UNCbar_0, 
\end{align*}
where $\bullet$ indicates the type of decomposition, namely, the Candille--Talagrand (\CT), the Brier score based (\BS), the quantile score based (\QS), or the modified Hersbach (\HB) decomposition.

\renewcommand{\arraystretch}{1.3}
\begin{table}[t] 
\centering
\caption{Candille--Talagrand (\CT), quantile score based (\QS), Brier score based (\BS), and Hersbach (\HB) decomposition of the mean score $\Sbar$, as applied to the one-day ahead raw ensemble (ENS) forecast of precipitation accumulation at Frankfurt Airport (Section \ref{sec:weather}), and the EasyUQ forecast for the Boston and Wine data, respectively (Section \ref{sec:ML}).}  \label{tab:decomposition}
\bigskip
\small
\begin{tabular}{lccccccc} 
\toprule
Forecast         & $\Sbar$ & $\UNCbar_0$ & $\MCBbar_\CT$ & $\MCBbar_\QS$ & $\MCBbar_\BS$ & $\MCBbar_\HB$  \\
\midrule
ENS              & 0.75    & 1.21        & 0.75          & 0.18          & 0.16          & 0.08 \\
EasyUQ (Boston)  & 1.75    & 4.76        & 1.75          & 0.72          & 0.57          & 0.36 \\
EasyUQ (Wine)    & 0.35    & 0.43        & 0.35          & 0.04          & 0.07          & 0.08 \\
\bottomrule
\end{tabular}
\end{table}

Table \ref{tab:decomposition} displays the mean score $\Sbar$, the uncertainty component $\UNC_0$, and the various $\MCBbar_\bullet$ terms for the ENS forecast of precipitation accumulation at Frankfurt Airport, as studied in our Section \ref{sec:weather} and \citet{IDR}, and the EasyUQ forecasts for the Boston Housing and Wine data, as considered in our Section \ref{sec:ML} and \citet{EasyUQ}.  The ENS forecast is an ensemble forecast with $m = 52$ members and so the Hersbach decomposition at \eqref{eq:Hersbach} applies; for the EasyUQ forecasts, we apply formula \eqref{eq:MCB_HB_finite_support} from Appendix \ref{app:proofs1}.  For the first two examples in the table, it holds that $\Sbar = \MCBbar_\CT > \MCBbar_\QS > \MCBbar_\BS > \MCBbar_\HB$, where the initial equality reflects the degeneracy of the Candille--Talagrand decomposition.  In our experience, the subsequent inequalities hold in many, though not all, empirical examples.  However, as we state in further generality at  \eqref{eq:inequalities_empirical} and in Corollary \ref{cor:order_of_MCB}, it always holds that $\Sbar \geq \MCBbar_\CT \geq \max \{ \MCBbar_\BS, \MCBbar_\QS \}$. 

While the Candille--Talagrand decomposition seems attractive and preferable from theoretical perspectives, the degeneracy prohibits its practical use.  The Hersbach decomposition has been popular in the specific setting of ensemble forecasts, but has serious shortcomings including but not limited to the possibility of a negative discrimination component.  The Brier score and quantile score based decompositions have desirable properties, but they define the components of the decomposition in terms of isolated functionals (probabilities and quantiles, respectively) rather than the entire predictive distributions, which is ``unsatisfactory'' \citep[p.~1958]{Ferro_Fricker_2012} and entails the artifacts described in Remarks \ref{rem:Fhat_i_need_not_be_cdfs} and \ref{rem:Fhat_i_need_not_be_cdfs_quantiles}, respectively.  Furthermore, it is not obvious whether the Brier score based or the quantile score based decomposition ought to be preferred.  In this light, there remains the need for a decomposition that is both practically feasible and theoretically justifiable and appealing. 

\section{Empirical isotonicity-based decomposition}  \label{sec:ISO}

We propose a method that builds on the idea of the Candille--Talagrand decomposition, but replaces auto-calibration with a slightly weaker notion of calibration, namely, isotonic calibration.  The resulting isotonicity-based decomposition, which we develop in this section, can be interpreted as a nondegenerate approximation to the Candille--Talagrand decomposition.

\subsection{Empirical isotonicity-based decomposition}  \label{sec:ISO_standard}

Recall that we denote by $\PR$ the class of the probability distributions on $\R$ with finite first moment.  For cdfs $F, G$, $F$ is stochastically smaller than or equal to $G$, for short $F \st G$, if $F(x)\geq G(x)$ for all $x \in \R$.  The stochastic order defines a partial order on $\PR$ and we refer to \citet{Shaked2007} for a comprehensive study. 

In the spirit of the Candille--Talagrand decomposition, a calibration tool ought to be applied to the assessed forecasts $F_1, \dots, F_n$ from \eqref{eq:data}, and we propose that this tool be isotonic distributional regression (IDR; \citealp{IDR}).  IDR is a nonparametric distributional regression method under the shape constraint of isotonicity between covariates and responses: For training data consisting of covariates $x_1, \dots, x_n$ in a partially ordered set $(\XX, \preceq)$ and real-valued responses $y_1, \dots, y_n$, \citet{IDR} prove that there exists a unique minimizer of the criterion
\begin{equation}  \label{eq:IDR_criterion}
\frac{1}{n} \sum_{i=1}^n \crps(P_i,y_i)
\end{equation}
over all vectors of {\cdf}s $(P_1, \dots, P_n)$ with $P_i \st P_j$ if $x_i \preceq x_j$ for $i, j = 1, \dots, n$, and they refer to this minimizer as the IDR solution.

The constraint of isotonicity between the assessed and the calibrated forecasts is natural, and hence, we apply IDR to the data $(F_1, y_1), \dots, (F_n,y_n)$ at \eqref{eq:data} with the stochastic order serving as the partial order on the covariate space $\PR$.  In a number of practically relevant situations the stochastic order is too strong, since it does not allow for crossings between {\cdf}s, and we discuss modifications that resolve this problem in the latter part of this section.  For now, we assume that there are sufficiently many pairs of cdfs across $F_1, \dots, F_n$ that can be ranked in stochastic order. 

Let $\check{F}_1, \dots, \check{F}_n$ denote the calibrated forecasts that are obtained by using IDR, let
\begin{align*}
\Sbar_\ISO = \frac{1}{n} \sum_{i=1}^n \crps(\check{F}_i,y_i)
\end{align*}
denote the mean score of the calibrated forecasts, let the marginal forecast $\hat{F}_\mg$ and its mean score $\Sbar_\mg$ be defined as at \eqref{eq:Sbar_mg_ac}, and let
\begin{align*}
\MCBbar_\ISO = \Sbar - \Sbar_\ISO, \quad \DSCbar_\ISO = \Sbar_\mg - \Sbar_\ISO.  
\end{align*}
Then the \textit{isotonicity-based} (ISO) decomposition
\begin{equation}  \label{eq:ISOdecomposition}
\Sbar = \MCBbar_\ISO - \DSCbar_\ISO + \UNCbar_0
\end{equation}
differs from the Candille--Talagrand decomposition at \eqref{eq:CTcomponents} by the choice of the calibration method only, as it draws on the slightly weaker notion of isotonic calibration in lieu of auto-calibration.  The isotonicity-based decomposition has desirable and appealing properties, as follows. 

\begin{prop}  \label{prop:ISO}
The isotonicity-based decomposition at \eqref{eq:ISOdecomposition} satisfies ($E_1$), ($E_2$), ($E_3$), ($E_4$), and ($E_5$).  Furthermore, $\MCBbar_\ISO = 0$ if, and only if, $F_i = \check{F}_i$ for\/ $i = 1, \dots, n$, and\/ $\DSCbar_\ISO = 0$ if, and only if, $\check{F}_i = \hat{F}_\mg$ for\/ $i = 1, \dots, n$.
\end{prop}

\begin{proof}
By definition, the isotonicity-based decomposition satisfies properties ($E_1$) and ($E_5$).  The IDR solution is the unique minimizer of the criterion \eqref{eq:IDR_criterion} over all vectors of distributions $(P_1, \dots, P_n)$ that are stochastically ordered with the same order relations as the covariates.  Here, the covariates are $F_1, \dots, F_n$ and the partial order on the covariate space is the stochastic order.  Therefore, $(F_1, \dots, F_n)$ is an admissible vector of distributions in the minimization problem, whence $\MCBbar_\ISO \ge 0$.  A further admissible vector in the minimization problem is the constant vector with entries $\hat{F}_\mg$, whence $\DSCbar_\ISO \ge 0$, so ($E_2$) is satisfied.  The examples in the case study in Section \ref{sec:case_study} imply that the isotonicity-based decomposition satisfies ($E_3$).  Assume now that $F_1 = \dots = F_n$.  Then we obtain $\hat{F}_\mg$ as the IDR solution, whence $\DSCbar_\ISO = 0$, so ($E_4$) is satisfied.  Finally, if $\MCBbar_\ISO = 0$ then $F_i = \check{F}_i$, since IDR is the unique minimizer of the criterion at \eqref{eq:IDR_criterion}, and analogously, if $\DSCbar_\ISO = 0$ then $\check{F}_i = \hat{F}_\mg$ for $i = 1, \dots, n$. 
\end{proof}

Generally, the determination of the pairwise stochastic order relations between the distributions $F_1, \ldots, F_n$ requires $\cO(n^2)$ operations.  As IDR can be implemented in at most $\cO(n^2)$ operations \citep{IDR, Henzi2022}, the computation of the isotonicity-based decomposition is of complexity ${\cal O}(n^2)$.  In contrast, the Brier score based and quantile score based decompositions require $\cO(n)$ or more distinct determinations of pairwise stochastic order relations (cf.~Appendices \ref{app:BS} and \ref{app:QS}) and, hence, the implementation is of complexity at least $\cO(n^2 \log n)$.  The computation of the Hersbach decomposition for an ensemble forecast of size $m$ requires $\cO(mn)$ operations.  
 
In its present form, the isotonicity-based decomposition is fully automated in the sense that it does not involve any tuning parameter.  For the examples in Table \ref{tab:decomposition}, $\MCBbar_\ISO$ equals $0.34$, $0.80$, and $0.072$, respectively, and so $\MCBbar_\ISO$ is larger than $\MCBbar_\BS$ (which equals 0.068 in the third example) and $\MCBbar_\QS$ and smaller than the essentially useless $\MCBbar_\CT = \Sbar$ term.  As we demonstrate in Section \ref{sec:properties}, it is always true that 
\begin{equation}  \label{eq:inequalities_empirical}
\Sbar \geq \MCBbar_\CT \geq \MCBbar_\ISO \geq \max \{ \MCBbar_\BS, \MCBbar_\QS \}.     
\end{equation}
In view of these theoretical guarantees in concert with its non-degeneracy and generality, we contend that the isotonicity-based method is more compelling than the Brier score or quantile score based decompositions. 

\subsection{Computational implementation}  \label{sec:ISO_approximate} 

When the predictive distributions are empirical distributions, stochastic order relations can be found by comparing the {\cdf}s at a finite number of real numbers, namely, the respective jump points.  If the predictive distributions are parametric, analytical results in terms of the parameters may be available; see, e.g., \citet{Shaked2007} and the proof of Proposition 1 in \citet{Gneiting2022}.

In relevant applications, the stochastic order may be to strong, since it allows for no crossings of the forecasts.  For example, for Gaussian forecasts $F = \mathcal{N}(\mu, \sigma^2)$ and $G = \mathcal{N}(\nu, \tau^2)$, $F$ and $G$ only order with respect to the stochastic order in case of $\sigma = \tau$, a condition which is rarely satisfied if parameters are estimated from data.  Generally, if $F$ and $G$ are members of a location-scale family, they are stochastically ordered if, and only if, they have equal scale parameter, subject to minimal conditions.  If only very few forecasts in the dataset are comparable with respect to the stochastic order, applying IDR results in calibrated forecast that are close to Dirac measures of the corresponding observations.  Hence, in principle, the isotonicity-based decomposition faces the same problem as the Candille--Talagrand decomposition in this setting.  However, we argue that there is a convincing remedy to the issue.

Consider settings where only few of the predictive distributions $F_i$ in the collection at \eqref{eq:data} are comparable with respect to the stochastic order.  Frequently, predictive distributions fail to order due to crossings of the {\cdf}s in a far tail.  Recent work by \citet{Brehmer2019} and \citet{Taillardat2023} casts doubt on the ability of the average \crps\ to distinguish tail behaviour of the forecast distribution, which provides support for the evaluation of the forecasts on a bounded interval only.  Motivated by these findings, instead of decomposing the original mean score $\Sbar$ as given in \eqref{eq:meanCRPS}, we decompose
\begin{equation}  \label{eq:crps_ab}
\Sbar^\ab = \frac{1}{n} \sum_{i=1}^n \crps(\tilde{F}_i^\ab, y_i),
\end{equation}
where for lower and upper threshold values $a \leq \min \{ y_1, \dots, y_n \}$ and $b \geq \max \{ y_1, \dots, y_n \}$, respectively, 
\begin{equation}  \label{eq:F_ab}
F_i^\ab(x) = 
\begin{cases}
0,      & x < a, \\
F_i(x), & x \in [a,b), \\
1,      & x \ge b, \\
\end{cases}
\end{equation}
for $i = 1, \ldots, n$.  Given an error tolerance $\epsilon > 0$, we determine the thresholds $a$ and $b$ such that the condition
\begin{equation}  \label{eq:bound}
\left\vert \, \Sbar - \Sbar^\ab \right\vert = \Sbar - \Sbar^\ab < \epsilon
\end{equation}
is satisfied, where the equality holds since $\Sbar \geq \Sbar^\ab$.  Condition \eqref{eq:bound} is equivalent to
\begin{align*}
I(a,b) = \frac{1}{n} \sum_{i=1}^n \left( \int_{-\infty}^a F_i(x)^2 \diff x + \int_b^\infty (1-F_i(x))^2 \diff x \right) < \epsilon.
\end{align*}
A simple method for determining the thresholds $a$ and $b$ to be used in \eqref{eq:F_ab} is described in Algorithm \ref{alg:iso_error_tol}.  If the support of the predictive distributions is bounded from above or below (e.g., in the case of precipitation accumulations, which are necessarily nonnegative), it is natural to set $a$ or $b$ equal to the respective bound (e.g., $a = 0$ for precipitation accumulations). 

\begin{algorithm}  
\caption{Thresholds $a, b$  \label{alg:iso_error_tol}}
\begin{algorithmic}[1]
    \State $\epsilon = \Sbar/1000$
    \State $a = \min \{ y_1, \dots, y_n \}$ and $b = \max \{ y_1, \dots, y_n \}$ 
    \If {$I(a,b) \geq \epsilon$}
        \State $\delta = (b - a)/100$
        \While {$I(a,b) \geq \epsilon$}
            \State $a = a - \delta$ and $b = b + \delta$
        \EndWhile
    \EndIf
    \State \textbf{return} $a, b$
\end{algorithmic}
\end{algorithm}

The computation of this modified isotonicity-based decomposition remains of complexity ${\cal O}(n^2)$.  Furthermore, the following result shows that, even with the approximation, theoretical guarantees from \eqref{eq:inequalities_empirical} continue to hold. 

\begin{prop}  \label{prop:inequalities}
Let\/ $\Sbar = \MCBbar_\ISO - \DSCbar_\ISO + \UNCbar_0 = \MCBbar_\BS - \DSCbar_\BS + \UNCbar_0$ denote decompositions for data $(F_1,y_1), \dots, (F_n,y_n)$, and let\/ 
\begin{align*}
\Sbar^\ab = \MCBbar_\ISO^\ab - \DSCbar_\ISO^\ab + \UNCbar_0 = \MCBbar_\BS^\ab - \DSCbar_\BS^\ab + \UNCbar_0
\end{align*}
denote the respective decompositions for modified data\/ $(F_1^\ab,y_1), \dots, (F_n^\ab,y_n)$, where\/ $F_1^\ab, \dots, F_n^\ab$ derive from\/ $F_1, \dots, F_n$ as in \eqref{eq:F_ab}.  Then\/ $I(a,b) = \Sbar - \Sbar^\ab < \epsilon$ implies that 
\begin{equation}  \label{eq:inequalities_approx}
\MCBbar_\ISO \ge \MCBbar_\ISO^\ab \geq \MCBbar_\BS^\ab > \MCBbar_\BS - \epsilon.
\end{equation}
\end{prop}

\begin{proof} 
The properties of the IDR solution imply $\Sbar_\ISO \le \Sbar_\ISO^\ab \le \Sbar^\ab \le \Sbar$, and we conclude that 
\begin{align*}
\MCBbar_\ISO = \Sbar - \Sbar_\ISO \ge \Sbar^\ab - \Sbar_\ISO^\ab = \MCBbar_\ISO^\ab.  
\end{align*}
To complete the proof, we apply the inequality \eqref{eq:inequalities_empirical} to the modified data to yield $\MCBbar_\ISO^\ab \ge \MCBbar_\BS^\ab$, and we note that $a \leq \min \{ y_1, \dots, y_n \}$ and $b \geq \max \{ y_1, \dots, y_n \}$, whence $\MCBbar_\BS - \MCBbar_\BS^\ab = I(a,b) < \epsilon$.
\end{proof} 

Assume that the predictive {\cdf}s belong to a location-scale family with full support, i.e., there exists a distribution $F_0 \in \PR$ with full support on $\R$ such that for $i = 1, \dots, n$ and $x \in \R$, $F_i(x) = F_0((x-\mu_i)/\sigma_i)$ for some location $\mu_i \in \R$ and scale $\sigma_i > 0$.  Then for any $i, j = 1, \dots, n$, the stochastic order relations between the modified distributions can be obtained based on the parameters \citep[proof of Proposition 1]{Gneiting2022}, in that 
\begin{align*}
F_i^\ab \st F_j^\ab
\end{align*}
if, and only if, $\mu_i \leq \mu_j$ and either $\sigma_i = \sigma_j$ or $(\mu_i \sigma_j - \mu_j \sigma_i)/(\sigma_j - \sigma_i) \notin [a,b]$.  In more complex but not uncommon situations, e.g., when the predictive distributions are mixtures of Gaussians, it may be hard to decide analytically whether or not there is a stochastic dominance relation between any two such distributions. A remedy is then to numerically evaluate and compare the {\cdf}s on a suitably chosen grid of threshold values.  As a default we suggest and use an equidistant grid from $a$ to $b$ of size 5000.  As long as the grid is sufficiently dense, order relations hardly ever change with the size of the grid, as experimental experience demonstrates. 

In order to increase the number of comparable pairs amongst $F_1, \dots, F_n$, it may appear natural to exchange the stochastic order with a weaker partial order on $\PR$, rather than restricting the support of the predictive distributions to a bounded interval $[a,b] \subseteq \R$.  However, we show in Appendix \ref{app:ISO} that isotonic calibration is generally only compatible with the stochastic order.  Therefore, the stochastic order is the only valid choice of a partial order if IDR is applied to generate a calibrated forecast for an isotonicity-based approach in the spirit of the Candille--Talagrand decomposition.  

\section{Population level analysis}  \label{sec:population}

In this section, we present population level versions of all decompositions which we have discussed so far, and we analyse their relations to notions of calibration.  The population quantity to be decomposed is the expected score 
\begin{equation}  \label{eq:Sbar_pop} 
\E \, \crps(F,Y),
\end{equation} 
where the expectation is with respect to the joint law $\prob$ of the random tuple $(F,Y)$ on a probability space $(\Omega, \F, \prob)$, where $F$ is a \cdf-valued random quantity, which we interpret as the forecast, and the random variable $Y$ is the real-valued outcome.  For subsequent use, we assume the existence of a standard uniform variable $U$ on $(\Omega, \F, \prob)$, which is independent of $(F,Y)$.  Evidently, if $\prob$ is the empirical distribution for the data at \eqref{eq:data} the expectation at \eqref{eq:Sbar_pop} reduces to the mean score $\Sbar$ from \eqref{eq:meanCRPS}.  
   
In all types of decompositions the population version of the uncertainty component is the expected score
\begin{equation}  \label{eq:theoretical_UNC}
\UNC_0 = \E \, \crps (F_\mg,Y)
\end{equation}
of the marginal law $F_\mg$ of $Y$.  Again, the expectation is with respect to $\prob$, and if $\prob$ is the empirical distribution of the data at \eqref{eq:data} then \eqref{eq:theoretical_UNC} reduces to \eqref{eq:UNC}.  In this light, the decompositions at the population level read 
\begin{align*}
\E \, \crps(F,Y) = \MCB_\bullet - \DSC_\bullet + \UNC_0,
\end{align*}
where $\bullet$ indicates the type, namely, \CT, \BS, \QS, \HB, or our new \ISO.  Therefore, it suffices to specify the miscalibration component $\MCB_\bullet$; the discrimination component is deduced as $\DSC_\bullet = \MCB_\bullet + \UNC_0 - \E \, \crps(F,Y)$.

\subsection{Desiderata for decompositions at the population level}  \label{sec:desiderata} 

We adapt the desirable properties ($E_1$) through ($E_5$) for decompositions of a mean score from Section \ref{sec:empirical} to the population setting, as follows. 
\begin{itemize} 
\item[($P_1$)] The decomposition is exact. 
\item[($P_2$)] The components $\MCB$, $\DSC$, and $\UNC$ are nonnegative.
\item[($P_3$)] The $\MCB$ component vanishes if, and only if, the forecast is calibrated in a well defined sense. 
\item[($P_4$)] The $\DSC$ component vanishes if the forecast is static, i.e., there is an $F_0 \in \PR$ such that $F = F_0$ almost surely.
\item[($P_5$)] The $\UNC$ component only depends on the unconditional distribution $F_\mg$ of the outcome.  
\end{itemize}

Concerning $(P_3)$, a notion of forecast calibration has to be specified.  In the special case of a binary outcome, there is a unique, clear-cut notion of calibration \citep[Theorem 2.11]{Gneiting_Ranjan_2013}.  Here, we consider the case of a real-valued outcome, for which numerous notions of calibration exist \citep{Tilmann_Johannes_Calibration}.  Auto-calibration is the strongest such notion, but typically cannot be used in practice.  Indeed, it turns out that $(E_3)$ and $(P_3)$ are competing requirements in the sense that if a decomposition satisfies $(P_3)$ with respect to auto-calibration, then $(E_3)$ is violated and the decomposition becomes degenerate.  If a weaker notion of calibration is requested for $(P_3)$, then $(E_3)$ can be satisfied for the empirical counterpart of the decomposition.  Requirement $(P_4)$ is natural, since a static forecast has no discrimination ability at all.  Finally, property $(P_5)$ is motivated by the observation that intrinsic uncertainty does not depend on the forecast; evidently, the criterion is satisfied by $\UNC_0$ at \eqref{eq:theoretical_UNC}.
 
\subsection{Isotonic conditional expectations and laws}  \label{sec:isotonic}

The population versions of the isotonicity-based, Brier score based, and quantile score based decompositions rely on conditional expectations given $\s$-lattices and isotonic conditional laws.  We give a short overview of the necessary concepts and refer to \cite{ICL} for further details.  Readers not familiar with measure theory might skip the current subsection and intuitively think of the conditional expectation and the conditional law of a random variable $Y$ given a $\s$-lattice $\A$, which we denote $\E( Y \mid \A)$ and $P_{Y \mid \A}$, respectively, as classical conditional expectations and laws under the constraint of isotonicity.  

Consider the probability space $(\Omega, \F, \prob)$.  A subset $\A \subseteq \F$ is a \textit{$\s$-lattice} if it is closed under countable unions and intersections and $\Omega,\emptyset \in \A$.  Let $\A \subseteq \F$ be a $\s$-lattice and let $X$ and $Z$ be integrable random variables defined on $(\Omega, \F, \prob)$.  We call $X$ \textit{$\A$-measurable} if $\{ X > x \} \in \A$ for all $x \in \R$ and define the \textit{$\s$-lattice generated by $X$}, denoted by $\LL(X)$, as the smallest $\s$-lattice which contains $\{ X > x \}$ for all $x\in \R$.  We call an $\A$-measurable random variable $\tilde{X}$ a \textit{conditional expectation of $X$ given $\A$}, for short $\E(X \mid \A)$, if $\E(X \one_A) \leq \E(\Tilde{X} \one_A)$ for all $A \in \A$ and $\E(X \one_B) = \E(\tilde{X} \one_B)$ for all $B \in \s(\tilde{X})$, where $\s(\tilde{X})$ denotes the $\s$-algebra generated by $\tilde{X}$.  \cite{Brunk_1965} showed that $\E(X \mid \A)$ is almost surely unique and coincides with the classical conditional expectations if $\A$ is a $\s$-algebra.  Conditional expectations given $\s$-lattices are closely connected to isotonicity as illustrated in \cite{ICL}.  In particular, for any integrable random variable $X$ and random variable $Z$, there exists an increasing Borel measurable function $f : \R \to \R$ such that $\E(X \mid \LL(Z)) = f(Z)$.  This result is analogous to the well-known factorization result for classical conditional expectations given $\s$-algebras, with the difference that, additionally, $f$ has to be increasing. 

Isotonic conditional laws can be defined in analogy to classical conditional laws.  Specifically, the isotonic conditional law (ICL) of the random variable $Y$ given $\A$, denoted $P_{Y \mid \A}$, is a Markov kernel from $(\Omega, \F)$ to $(\R, \BR)$ such that $\omega \mapsto P_{Y \mid \A}(\omega,(y, \infty))$ is a version of $\prob(Y > y \mid \A) = \E(\one \{Y > y\} \mid \A)$ for any $y \in \R$.  \cite{ICL} show the existence and uniqueness of ICL.  Equivalently, ICL emerges as the minimizer of an expected score, where the scoring rule may be taken from a large class of proper scoring rules that includes the \crps.  

We are particularly interested in ICL with respect to the $\s$-lattice generated by the forecast $F$.  We call $B\subseteq \PR$ an upper set if $P \in B$ and $P \st Q$ implies $Q \in B$ for $Q \in \PR$, and we denote by $\mathcal{U}$ the family of all upper sets in $\PR$.  For the forecast $F$, we define the $\s$-lattice generated by $F$ as the family of all preimages of measurable upper sets under $F$, i.e., $\LL(F) = \big \{ F^{-1}(B) \mid B \in \BPR \cap \mathcal{U} \} \subseteq \F$, where $\BPR$ denotes the $\s$-algebra on $\PR$ with respect to the weak topology.  For details, we refer the reader to Definition 3.1 of \citet{ICL}.

In a nutshell, $P_{Y \mid \LL(F)}$ arises as the best available prediction for the distribution of $Y$, given all information in the forecast $F$, under the assumption that smaller (greater) values of $F$ correspond to smaller (greater) values of the conditional law with respect to the stochastic order.  

\subsection{Calibration}  \label{sec:calibration}

A strong notion of calibration is auto-calibration, which formalizes the idea that the outcome is indistinguishable from a random draw from the posited distribution $F$.  Specifically, the random forecast $F$ is \textit{auto-calibrated} \citep{Tsyplakov2013} if $P_{Y \mid F} = F$, or equivalently
\begin{equation}  \label{defn_auto_calibration}
F(x) = \prob(Y \leq x \mid F) \quad \textrm{almost surely for all } x \in \R.
\end{equation}
For any threshold value $x \in \R$, we may condition on the random variable $F(x)$ instead of the random distribution $F$ in \eqref{defn_auto_calibration}, to obtain the weaker notion of threshold calibration.  Specifically, the forecast $F$ is called \textit{threshold calibrated} \citep{IDR} if
\begin{align*}
F(x) = \prob(Y \leq x \mid F(x)) \quad \textrm{almost surely for all } x \in \R.
\end{align*}
Essentially, for a threshold calibrated forecast $F$, we can take $F(x)$ at face value for any $x \in \R$.  In a slight adaptation of the definition in \cite{Tilmann_Johannes_Calibration}, we call the forecast $F$ \textit{quantile calibrated} if
\begin{align*}
F^{-1}(\alpha) = q_\alpha(Y \mid F^{-1}(\alpha)) \quad \textrm{almost surely for all } \alpha \in (0,1), 
\end{align*}
where for any $\alpha \in (0,1)$, $q_\alpha(Y \mid F^{-1}(\alpha))$ denotes the lower-$\alpha$-quantile of the conditional law of $Y$ given $F^{-1}(\alpha)$.  Equivalently, one can think of $q_\alpha(Y \mid F^{-1}(\alpha))$ as a $\s(F^{-1}(\alpha))$-measurable random variable which minimizes $\E \, \qs(G,Y)$ over all $\s(F^{-1}(\alpha))$-measurable random variables $G$; see \cite{Armerin}.

The forecast $F$ is called \textit{isotonically calibrated} if $F$ is almost surely equal to the isotonic conditional law of $Y$ given $\LL(F)$, i.e., $F = P_{Y \mid \LL(F)}$ almost surely.  By Proposition 5.3 of \citet{ICL}, auto-calibration implies isotonic calibration, and isotonic calibration implies threshold calibration and quantile calibration.  

\begin{figure}[t]
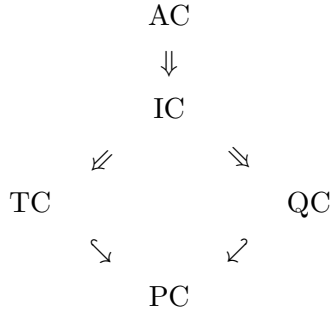

\centering
\begin{tabular}{ccccc}
   &            & AC           &             & \\ 
   &            & $\Downarrow$ &             & \\
   &            & IC           &             & \\
   & $\SWarrow$ &              & $\SEarrow$  & \\
TC &            &              &             & QC \\
   & $\SEArrow$ &              & $\NEArrows$ &   \\
   &            & PC           &             &
 \end{tabular}
\caption{Implications between auto-calibration (AC), isotonic calibration (IC), threshold calibration (TC), and quantile calibration (QC).  Implications with respect to probabilistic calibration (PC) are indicated by hooked arrows and hold under Assumption 2.15 of \citet{Tilmann_Johannes_Calibration}.  \label{fig:calibration}}
\end{figure}

The probability integral transform (PIT) of the \cdf-valued random quantity $F$ is the random variable $Z_F = F(Y-) + U (F(Y)-F(Y-))$, where $F(y-) = \lim_{x \uparrow y} F(x)$ denotes the left-hand limit of $F$ at $y \in \R$, with a random variable $U$ that is standard uniform and independent of $F$ and $Y$.  The PIT of a continuous \cdf\ $F$ simplifies to $Z_F = F(Y)$.  The forecast $F$ is \textit{probabilistically calibrated} if $Z_F$ is uniformly distributed on the unit interval \citep{Gneiting_Ranjan_2013}.  Originally suggested by \cite{Dawid1984}, checks for probabilistic calibration, and for the uniformity of the closely related rank histogram, constitute a cornerstone of forecast evaluation \citep{Diebold1998, Hamill_2001, Gneiting2007b}.  Under regularity conditions, a threshold calibrated or quantile calibrated forecast is probabilistically calibrated; details and a direct implication from isotonic calibration to a weak form of probabilistic calibration are available in \citet[Section 3.3]{Tilmann_Johannes_Calibration} and \citet[Appendix D]{ICL}, respectively.  Figure \ref{fig:calibration} summarizes relationships between the notions of calibration discussed in this section. 

\subsection{Population level decompositions}  \label{sec:decomposition}

We now give generalizations of the empirical decompositions discussed in Sections \ref{sec:empirical} and \ref{sec:ISO} that apply at the population level.  Recall that we consider the joint law $\prob$ of the random tuple $(F,Y)$.  As before, we let $\PR$ denote the class of the Borel probability measures on $\R$ that have a finite first moment.  In the current and the subsequent subsection, we generally operate under the following regularity conditions.  For proofs, we refer the reader to Appendix \ref{app:proofs1}. 

\begin{assump}  \label{assump:population}
Let the marginal law $F_\mg$ of $Y$ be such that $F_\mg \in \PR$, and suppose that
\begin{equation}  \label{eq:EEXdF}
\E \int \lvert x \rvert \diff F(x) < \infty.
\end{equation}
\end{assump}

In view of the kernel score representation of the \crps\ \citep[eq.~(21)]{Gneiting2007a}, Assumption \ref{assump:population} implies that
\begin{align*}
\E \, \crps(F,Y)
& = \E \, \E (\crps(F,Y) \mid F ) \\
& = \E \left( \E_F (\lvert X - Y \rvert \mid F) - \frac{1}{2} \E_F (\lvert X - X' \rvert \mid F) \right) \\
& \leq \E \, \E_F \lvert X \rvert + \E \, \lvert Y \rvert < \infty,
\end{align*}
where $X$ and $X'$ are independent random variables with law $F$.  Similarly, it follows that $\E \, \crps (F_\mg,Y) < \infty$.  Furthermore, the properties of isotonic and standard conditional laws imply that $\E \, \crps (P_{Y \mid \LL(F)},Y) \leq  \E \, \crps (F,Y)$ and $\E \, \crps (P_{Y \mid F},Y) \leq  \E \, \crps (F,Y)$, respectively.  In this light,  Assumption \ref{assump:population} ensures that $\E \, \crps(F,Y)$, $\E \, \crps (F_\mg,Y)$, $\E \, \crps (P_{Y \mid \LL(F)},Y)$, and $\E \, \crps (P_{Y \mid F},Y)$ are finite.

The population version of the Candille--Talagrand decomposition at \eqref{eq:CTdecomposition} is
\begin{equation}  \label{eq:CT}
\E \, \crps(F,Y) = \MCB_\CT - \DSC_\CT + \UNC_0,
\end{equation}
where $\UNC_0$ is defined at \eqref{eq:theoretical_UNC}, and
\begin{align*}
\MCB_\CT = \E \, \crps(F, Y) - \E \, \crps( P_{\, Y \mid F}, Y).
\end{align*}
Similarly, the population version of the isotonicity-based decomposition at \eqref{eq:ISOdecomposition} is
\begin{equation}  \label{eq:ISO}
\E \, \crps(F,Y) = \MCB_\ISO - \DSC_\ISO + \UNC_0,
\end{equation}
where 
\begin{equation*}
\MCB_\ISO = \E \, \crps(F,Y) - \E \, \crps( P_{Y \mid \LL(F)}, Y).
\end{equation*}
The decomposition at \eqref{eq:ISO} is analogous to the theoretically preferred Candille--Talagrand decomposition at \eqref{eq:CT}, except that the performance of the forecast $F$ is compared with the isotonic conditional law $P_{Y \mid \LL(F)}$ rather than the conditional law $P_{Y \mid F}$.  
The general decompositions at \eqref{eq:CT} and \eqref{eq:ISO} reduce to \eqref{eq:CTdecomposition} and \eqref{eq:ISOdecomposition}, respectively, when $\prob$ is the empirical distribution of the data in \eqref{eq:data}.   
 
The population version of the Brier score based decomposition at \eqref{eq:BSempirical} is 
\begin{equation}  \label{eq:BSpopulation}
\E \, \crps(F,Y) = \MCB_\BS - \DSC_\BS + \UNC_0,
\end{equation}
where
\begin{align*}
\MCB_\BS = \E \, \crps(F,Y)  - \E \int \big( \prob(Y \leq z \mid \LL(F(z))) - \one \{ Y \leq z \} \big)^2 \diff z.
\end{align*}
Similarly, the population version of the quantile based based decomposition at \eqref{eq:QSempirical} is  
\begin{equation}  \label{eq:QSpopulation}
\E \, \crps(F,Y) = \MCB_\QS - \DSC_\QS + \UNC_0,
\end{equation}
where
\begin{align*}
\MCB_\QS = \E \, \crps(F,Y) - \E \int_0^1 \qs \big( q_\alpha(Y \mid \LL(F^{-1}(\alpha))), Y \big) \diff \alpha.
\end{align*}
The properties of isotonic conditional expectations and isotonic conditional quantiles imply that $ \E \int ( \prob(Y \leq z \mid \LL(F(z))) - \one \{ Y \leq z \} )^2 \diff z\leq \E \, \crps(F,Y)  < \infty$ and $ \E \int_0^1 \qs ( q_\alpha(Y \mid \LL(F^{-1}(\alpha))), Y ) \diff \alpha \leq \E \, \crps(F,Y)  < \infty$.  The decompositions at \eqref{eq:BSpopulation} and \eqref{eq:QSpopulation} reduce to \eqref{eq:BSempirical} and  \eqref{eq:QSempirical}, respectively, when $\prob$ is the empirical distribution of the data in \eqref{eq:data}.

Finally, we consider the Hersbach decomposition.  To this end, let $\nu_F$ be the image of the Lebesgue measure $\lambda$ under $F$, i.e., $\nu_F(A) = \lambda(F^{-1}(A))$, and define the measures given by  
\begin{equation}  \label{eq:mu}
\mu(A) = \E \left( \nu_F(A) \right)
\end{equation}
and 
\begin{equation}  \label{eq:tau}
\tau(A) = \E \left( \int_A \one \{ F(Y) \leq p \} \diff \nu_F(p) \right) \! ,  
\end{equation}
respectively, where $A \in \BU$ is any Borel set.  We are now ready to state a population version of the Hersbach decomposition from Section \ref{sec:HB}.

\begin{prop}  \label{prop:HB_general}
Let Assumption \ref{assump:population} hold, and let\/ $\mu$ and\/ $\tau$ be the measures defined at \eqref{eq:mu} and \eqref{eq:tau}, respectively.  Then\/ $\tau$ is absolutely continuous with respect to\/ $\mu$; let\/ $f$ denote the respective Radon--Nikodym derivative.  It holds that 
\begin{equation}  \label{eq:HB}
\E \, \crps(F,Y) = \MCB_\HB - \DSC_\HB + \UNC_0,
\end{equation}
where\/ $\UNC_0$ is given at \eqref{eq:theoretical_UNC}, 
\begin{align*}
\MCB_\HB = \int_0^1 (p-f(p))^2 \, \diff \mu(p), \quad \DSC_\HB = \UNC_0 - \int_0^1 f(p) (1-f(p)) \diff \mu(u) - \MS,
\end{align*}
and
\begin{equation}  \label{eq:MS_Hersbach}
\MS = \E \big[ \one \{ F(Y) = 0 \} \, (F^{-1}(0+) - Y) + \one \{ F(Y) > 0 \} \, (2F(Y) -1)(Y - F^{-1}(F(Y))) \big]. 
\end{equation}
\end{prop}

The \MS\ component can only be nonzero when $Y$ lies outside the support of $F$ with positive probability; hence, we write MS for misspecified support.  Note that \MS\ can be negative, e.g., if $F = (\delta_0 + 3 \, \delta_2)/4$ and $Y = 1$ almost surely then $\MS = - 1/2$.

The following result is a special case of the more general statement in Corollary \ref{cor:HB_for_step_functions} in Appendix \ref{app:proofs1}.  It shows that the population decomposition nests the modified empirical Hersbach decomposition.  

\begin{cor}  \label{cor:HB_2}
If\/ $\prob$ is the empirical measure of a collection of forecast--observation pairs $(F_1,y_1), \dots, (F_n,y_n)$, where each $F_i$ is the empirical cdf of a sample of size $m$, then the population decomposition at \eqref{eq:HB} reduces to the modified empirical Hersbach decomposition at \eqref{eq:Hersbach_mod}. 
\end{cor}

The next result demonstrates that Proposition \ref{prop:HB_general} subsumes the Hersbach--Lalaurette decomposition for strictly increasing forecast cdfs as given in Appendix A of \citet{Candille_2005}.

\begin{cor}  \label{cor:HB_1}
Let Assumption \ref{assump:population} hold, and suppose that\/ $F^{-1}$ is almost surely absolutely continuous.  Then\/ $\MS = 0$ and the measure\/ $\mu$ at \eqref{eq:mu} has density 
\begin{equation}  \label{eq:gamma}
\gamma(p) = \E \left( \frac{\diff}{\diff p} F^{-1}(p) \right)
\end{equation}
with respect to the Lebesgue measure on the unit interval.  Furthermore,  the measure\/ $\tau$ at \eqref{eq:tau} has Radon--Nikodym derivative defined by 
\begin{equation}  \label{eq:f}
f(p) = \frac{1}{ \gamma(p)} \, \E \left( \one \{ F(Y) \leq p \} \, \frac{\diff}{\diff p} F^{-1}(p) \right)
\end{equation}
if\/ $\gamma(p) > 0$, and\/ $f(p) = 0$ otherwise, with respect to $\mu$.
\end{cor}

Considering a practically relevant case, we derive in Example \ref{ex:prec_forecasts} in Appendix \ref{app:proofs1} the empirical Hersbach decomposition for probabilistic forecasts of a nonnegative quantity, assuming that the forecast distributions are mixtures of a point mass at zero and a strictly positive density on the positive halfline.

\subsection{Properties of the decompositions}  \label{sec:properties} 
 
The population versions of the Candille--Talagrand, isotonicity-based, Brier score based, and quantile score based decompositions satisfy properties ($P_1$), ($P_2$), ($P_4$), and ($P_5$), and property ($P_3$) with auto-calibration, isotonic calibration, threshold calibration, and quantile calibration, respectively.  The following theorem and its proof summarize and elaborate on property ($P_3$) and lend theoretical support to the use of the isotonicity-based decomposition.  While in principle one would like to quantify miscalibration in terms of deviations from auto-calibration, as done by the Candille--Talagrand decomposition, the empirical version thereof is degenerate.  By imposing the natural shape constraint of isotonicity between the assessed and the calibrated forecasts, a practically useful decomposition is obtained that does not rely on implementation choices, save for a possible choice of threshold values $a$ and $b$ in the modified cdfs $F^\ab$ at \eqref{eq:F_ab}.  The isotonicity-based decomposition quantifies miscalibration as deviation from isotonic calibration, which is closer to auto-calibration than threshold or quantile calibration as illustrated in Figure \ref{fig:calibration}.

All proofs for this section are deferred to Appendix \ref{app:proofs2}. 

\begin{thm}  \label{thm:calibration}
Under Assumption \ref{assump:population} the following statements hold.
\begin{itemize}
\item[(a)]  
The Candille--Talagrand decomposition at \eqref{eq:CT} is exact and satisfies 
\begin{itemize}
\item $\MCB_\CT \geq 0$ with equality if, and only if, $F$ is auto-calibrated;
\item $\DSC_\CT \geq 0$ with equality if, and only if, $P_{Y \mid F} = F_\mg$ almost surely.
\end{itemize}
\item[(b)]  
The isotonicity-based decomposition at \eqref{eq:ISO} is exact and satisfies 
\begin{itemize}
\item $\MCB_\ISO \geq 0$ with equality if, and only if, $F$ is isotonically calibrated;
\item $\DSC_\ISO \geq 0$ with equality if, and only if, $P_{Y \mid \LL(F)} = F_\mg$ almost surely.
\end{itemize}
\item[(c)]  
The Brier score based decomposition at \eqref{eq:BSpopulation} is exact and satisfies 
\begin{itemize}
\item $\MCB_\BS \geq 0$ with equality if, and only if, $F$ is threshold calibrated;
\item $\DSC_\BS \geq 0$ with equality if, and only if, for all\/ $z\in \R$, $\prob( Y \leq z \mid \LL(F(z)))$ $= \prob(Y \leq z)$ almost surely. 
\end{itemize}
\item[(d)]  
The quantile score based decomposition at \eqref{eq:QSpopulation} is exact and satisfies 
\begin{itemize}
\item $\MCB_\QS \geq 0$ with equality if\/ $F$ is quantile calibrated; conversely, if the random element\/ $(Y,F^{-1}(\alpha))$ satisfies Assumption 6.1 in \cite{ICL} for all\/ $\alpha \in (0,1)$ then\/ $\MCB_\QS = 0$ implies quantile calibration of\/ $F$; 
\item $\DSC_\QS \geq 0$ with equality if, and only if, for all\/ $\alpha \in (0,1)$, $q_\alpha(Y \mid \LL(F^{-1}(\alpha)))$ $= q_\alpha(Y)$ almost surely. 
\end{itemize}
\end{itemize}
\end{thm}

In view of known relationships between notions of calibration \citep[Sections 2.2 and 2.3]{Tilmann_Johannes_Calibration} the following implications hold.

\begin{cor}  \label{cor:auto_MCB}
Under Assumption \ref{assump:population}, an auto-calibrated forecast yields\/ $\MCB_\CT = \MCB_\ISO = \MCB_\BS = \MCB_\QS = 0$. 
\end{cor}

\begin{cor}  \label{cor:order_of_MCB}
Under Assumption \ref{assump:population}, it holds that 
\begin{equation}  \label{eq:inequalities}
\E \, \crps(F,Y) \geq \MCB_\CT \geq \MCB_\ISO \geq \max \{ \MCB_\BS, \MCB_\QS \}.
\end{equation}
\end{cor}

Importantly, while formulated at the population level, the above results apply to the empirical versions of the decompositions, by identifying the joint distribution $\prob$ of the tuple $(F,Y)$ with the empirical law of the data at \eqref{eq:data}.  In particular, the relations in \eqref{eq:inequalities} nest the respective inequalities \eqref{eq:inequalities_empirical} for the empirical decompositions.  For the isotonicity-based decomposition, if modified {\cdf}s $F^\ab$ are used the results apply to the latter, and we refer to \eqref{eq:inequalities_approx} for relationships to the respective components computed on the original {\cdf}s.  

Finally, we consider the Hersbach decomposition from Proposition \ref{prop:HB_general}, which struggles to satisfy the desirable properties from Section \ref{sec:desiderata}.  By definition, properties ($P_{1}$) and ($P_{5}$) hold.  The miscalibration component is clearly nonnegative.  However, $\DSC_\HB$ may be negative as in Example \ref{app:3}, i.e., property ($P_{2}$) is violated.  Moreover, the example in the proof of Proposition \ref{prop:HB_decomp} shows that the Hersbach decomposition fails to satisfy ($P_4$).   Concerning ($P_3$), \cite{Hersbach_2000} and \cite{Candille_2005} argue that the Hersbach reliability component is closely related to the rank histogram and hence one might expect that $\MCB_\HB = 0$ if, and only if, $F$ is probabilistically calibrated.  However, the examples in Appendices \ref{app:4} and \ref{app:5} show that probabilistic calibration is neither sufficient nor necessary for $\MCB_\HB = 0$.  The following proposition collects calibration properties in relation to the Hersbach decomposition.

\begin{prop}  \label{prop:properties_Hersbach_decomp} 
Let Assumption \ref{assump:population} hold and consider the population version of the Hersbach decomposition at \eqref{eq:HB}.
\begin{itemize}
\item[(a)] If\/ $Y \in \supp(F)$ almost surely, then\/ $\MS = 0$, where\/ $\MS$ is defined at \eqref{eq:MS_Hersbach}.
\item[(b)] For an auto-calibrated forecast, it holds that\/ $\MS = \MCB_\HB = 0$.
\item[(c)] Suppose that\/ $F$ belongs to a location family, i.e., for all\/ $x \in \R$, $F(x) = F_0(x - \mu)$ for some\/ $F_0 \in \PR$ and random location\/ $\mu$.  Suppose furthermore that\/ $F_0$ has no jumps and\/ $F_0^{-1}$ is absolutely continuous.  Then\/ $\MCB_\HB = 0$ if\/ $F$ is probabilistically calibrated.
\end{itemize}
\end{prop}

\begin{figure}[t]
\centering
\includegraphics[width = 0.60\textwidth]{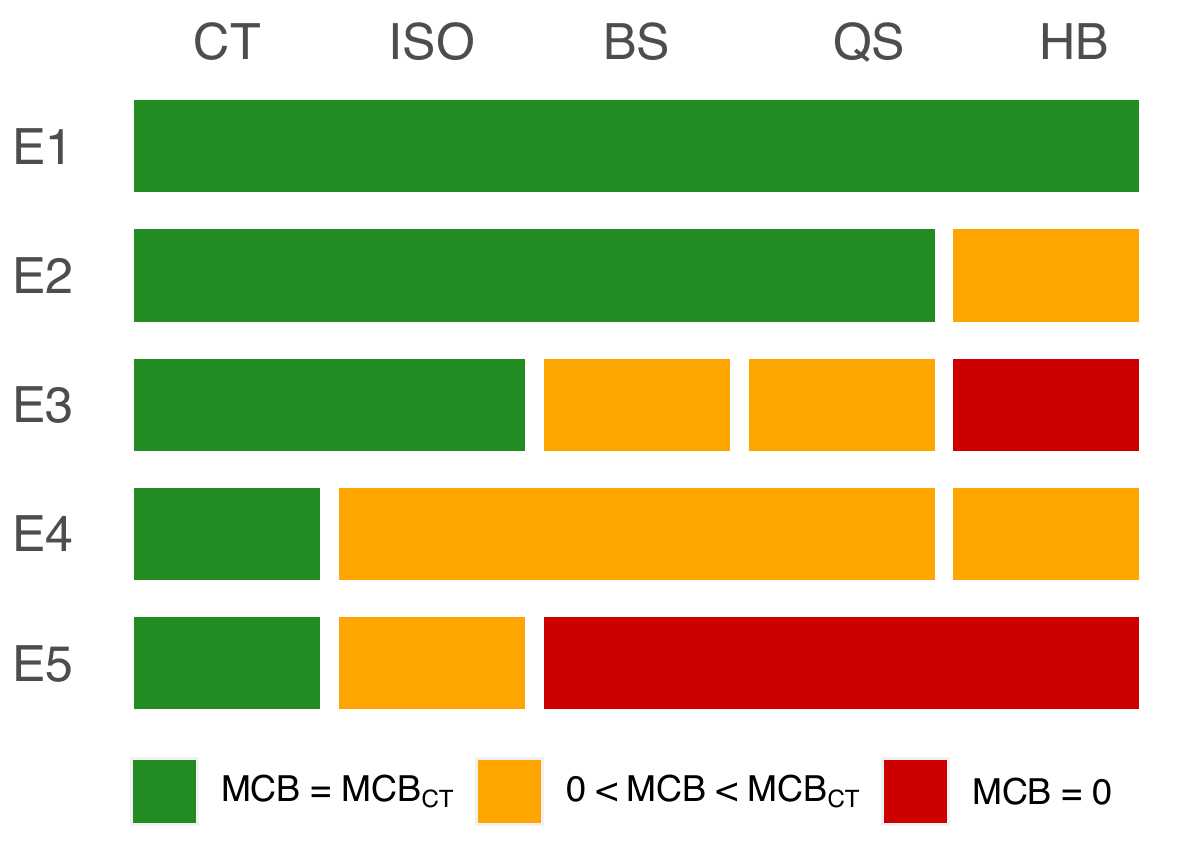}
\caption{The graphic indicates for the population level examples E1, \ldots, E5 in Appendix \ref{app:examples} whether the $\MCB_\bullet$ term, where $\bullet$ stands for \CT, \ISO, \BS, \QS, or \HB, respectively, agrees with the theoretically preferred quantity $\MCB_\CT$ (green), is smaller than $\MCB_\CT$ but remains positive (orange), or deceptively equals zero (red).  Connected segments indicate equality of corresponding terms.  For analytic results, see Table \ref{tab:examples}.  \label{fig:examples}}
\end{figure}

In Appendix \ref{app:examples} we compare the different types of decompositions in a number of analytic examples at the population level.  Figure \ref{fig:examples} summarizes how the respective miscalibration terms relate to the theoretically preferred $\MCB_\CT$ component.

\section{Case studies}  \label{sec:case_study}

We now illustrate the use of the isotonicity-based decomposition from Section \ref{sec:ISO} in case studies on weather forecasts and benchmark regression tasks from machine learning, respectively.  For simplicity, we use an abbreviated notation for the components of the mean score $\Sbar$ throughout this section, namely, $\MCBbar = \MCBbar_\ISO$, $\DSCbar = \DSCbar_\ISO$, and $\UNCbar = \UNCbar_0$, respectively.  Note the opposite orientation of $\MCBbar$ and $\DSCbar$, in that higher $\DSCbar$ corresponds to better discrimination ability, whereas lower $\MCBbar$ indicates better calibration. 

When one seeks to simultaneously compare $\Sbar$, $\MCBbar$, and $\DSCbar$ between larger numbers of forecast methods, tables get cumbersome.  Therefore, we suggest a graphical display, namely, the $\MCBbar$--$\DSCbar$ plot, which is motivated by similar displays in \citet{Dimitriadis_2023} and \citet{Gneiting_2023}.  In this type of graphic, $\MCBbar$ is plotted against $\DSCbar$, and isolines correspond to specific values of the mean score $\Sbar$, which is constant along parallel lines.  The uncertainty component $\UNCbar$ is independent of the forecast method, and we display it in the upper left or upper right corner of the plot.

\subsection{Probabilistic quantitative precipitation forecasts}  \label{sec:weather}

Ensemble prediction systems have tremendously improved weather forecasts over the past decades \citep{Bauer2015}.  However, ensemble forecasts remain subject to biases and dispersion errors, and hence require some form of statistical postprocessing \citep{Gneiting2005a, Vannitsem2018}.   Here we consider the case study in \cite{IDR}, which compares the performance of raw and postprocessed ensemble forecasts for 24-hour accumulated precipitation in terms of the mean score $\Sbar$, which we decompose into $\MCBbar$, $\DSCbar$, and $\UNCbar$, respectively.  

Following \cite{IDR}, we consider forecasts and observations for 24-hour accumulated precipitation from 6 January 2007 to 1 January 2017 at Brussels, Frankfurt, London, and Zurich in millimeters.  The 52 member raw ensemble (ENS) forecast operated by the European Centre for Medium-Range Weather Forecasts comprises a high resolution member, a control member at lower resolution, and 50 perturbed members at the same lower resolution but with perturbed initial conditions \citep{Molteni1996}.  We use data from 2007 to 2014 to train the postprocessing techniques Bayesian model averaging \citep[BMA;][]{Sloughter2007}, ensemble model output statistics \citep[EMOS;][]{Scheuerer2014}, heteroscedastic censored logistic regression \citep[HCLR;][]{Messner2014} and two versions, IDR$_\textrm{cw}$ and IDR$_\textrm{st}$, of isotonic distributional regression \citep[IDR;][]{IDR}, where IDR$_\textrm{cw}$ is documented in \cite{IDR} and IDR$_\textrm{st}$ uses the stochastic order on the ensemble {\cdf}s.  For further implementation details we refer the reader to \citet{IDR}.  The years 2015 and 2016 form the evaluation period.  

\begin{figure}[p]
\centering
\includegraphics[width = 0.49\textwidth]{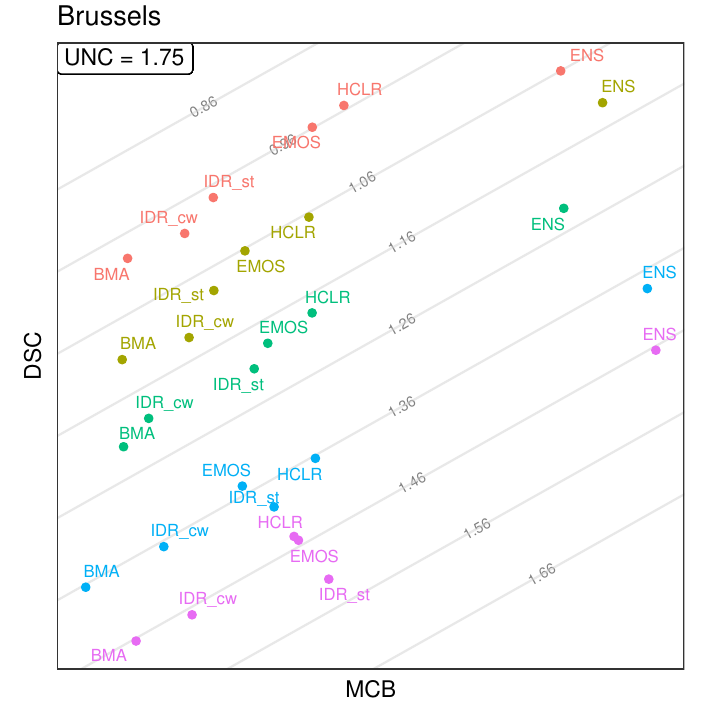}
\includegraphics[width = 0.49\textwidth]{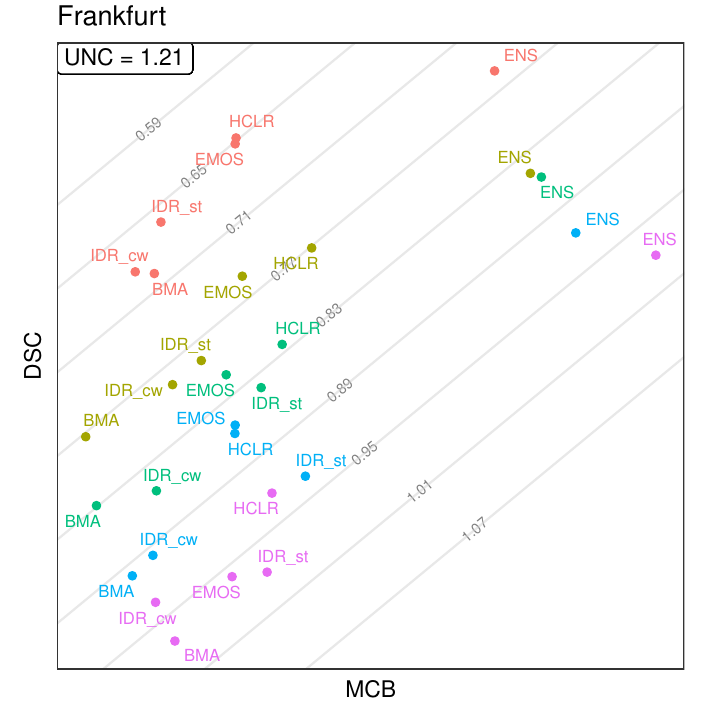}
\includegraphics[width = 0.49\textwidth]{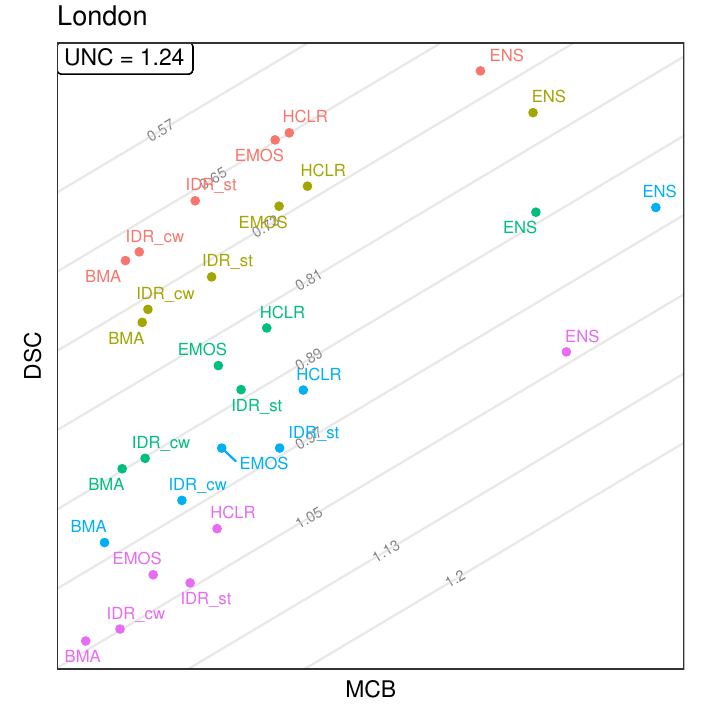}
\includegraphics[width = 0.49\textwidth]{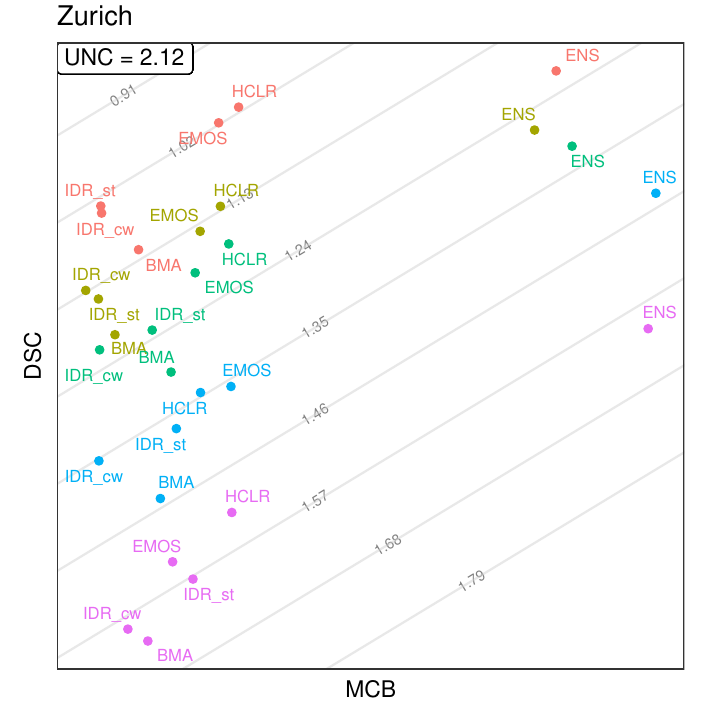}
\includegraphics[width = 0.49\textwidth]{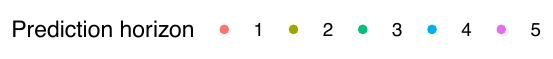}
\caption{$\MCBbar$--$\DSCbar$ plots for forecasts of 24-hour accumulated precipitation at Brussels, Frankfurt, London, and Zurich, at prediction horizons of one to five days ahead.  The mean score $\Sbar$ is constant along the parallel lines and shown in the unit of millimeters.  Acronyms are defined in the text, and details of the forecast methods are documented in \citet[Section 5]{IDR}.  \label{fig:weather}}
\end{figure}

The ENS and IDR forecast distributions have finite support and we apply the isotonicity-based decomposition of $\Sbar$ in its pure form from Section \ref{sec:ISO_standard}.  For the other forecasts, which employ mixtures of a point mass at zero (for no precipitation) and a density at positive accumulations as predictive distributions, we fix $a = 0$ and use Algorithm \ref{alg:iso_error_tol} to determine the upper bound $b$, which generally is identical to, or very slightly higher than, the highest accumulation observed in the test data; then we compute stochastic order relations on an equidistant grid of size $5000$ over $[a,b]$ and apply the isotonicity-based decomposition in its approximate form from Section \ref{sec:ISO_approximate}. 

The respective $\MCBbar$--$\DSCbar$ plots for Brussels, Frankfurt, London, and Zurich are shown in Figure \ref{fig:weather}.  We note an increase of the mean score $\Sbar$ values with the prediction horizon, which is due to a decrease in discrimination ability.  The raw ensemble (ENS) forecasts discriminate very well, but are poorly calibrated.  The postprocessing methods yield considerable improvement in $\Sbar$, subject to a trade-off between $\MCBbar$ and $\DSCbar$.  The EMOS and HCLR techniques, which employ inflexible parametric densities with fixed shape, excel in terms of discrimination, but lack in calibration.  In contrast, the BMA and IDR techniques, which are much more flexible, are better calibrated, but inferior in terms of discrimination ability.

\subsection{Benchmark regression problems from machine learning}  \label{sec:ML}  

A sizable strand of recent literature in machine learning is concerned with methods for uncertainty quantification for neural networks, where the task is the transformation of single-valued neural network output into predictive distributions \citep{Gawlikowski2023}.  In this literature, performance is typically evaluated in terms of the mean logarithmic score \citep[Section 4.1]{Gneiting2007a} which, in sharp contrast to the crps, can only be applied to methods that generate predictive densities.  Furthermore, extant measures for the assessment of calibration and discrimination ability tend to be ad hoc. In this section, we demonstrate the use of the mean score $\Sbar$ and its isotonicity-based decomposition into $\MCBbar$, $\DSCbar$, and $\UNCbar$ in this context. 

We adopt the benchmark regression tasks setting originally proposed by \cite{HernandezLobato2015} and consider the datasets and methods from the middle block of Table 6 in \citet{EasyUQ}, except that we skip results for the Naval and Year datasets, for which there are missing entries.  The experimental setting is based on single-valued output from a neural network, which learns a regression function based on a collection of covariates or features.  In this setting, \cite{Walz2022} compare competing methods for uncertainty quantification, including the popular Monte Carlo Dropout approach \citep[MC Dropout;][]{Gal2016} and a scalable Laplace approximation based technique \citep[Laplace;][]{Immer2021, Ritter2018} that operate within the neural network learning pipeline.  Their competitors include output-based methods that learn on training data of previous single-valued model output and outcomes only, without accessing feature values, namely, the Single Gaussian technique, conformal prediction \citep[CP;][]{Vovk2020}, and the EasyUQ technique \citep{Walz2022}, which is based on IDR \citep{IDR}.  Furthermore, we consider smoothed versions of the discrete CP and EasyUQ distributions, termed Smooth CP and Smooth EasyUQ, respectively.  For implementation details, we refer the reader to \citet{Walz2022}.  

\begin{sidewaysfigure}[p]
\centering
\includegraphics[width = 0.245\textwidth]{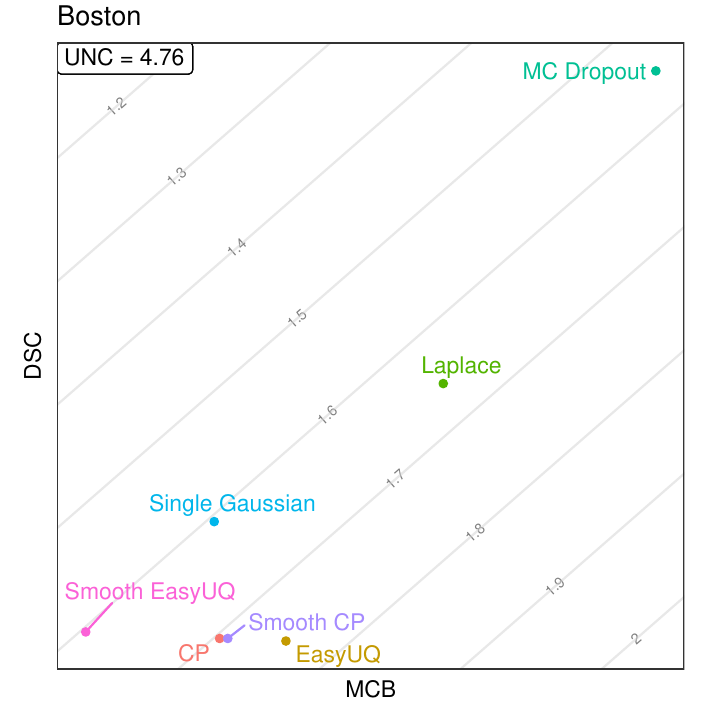}
\includegraphics[width = 0.245\textwidth]{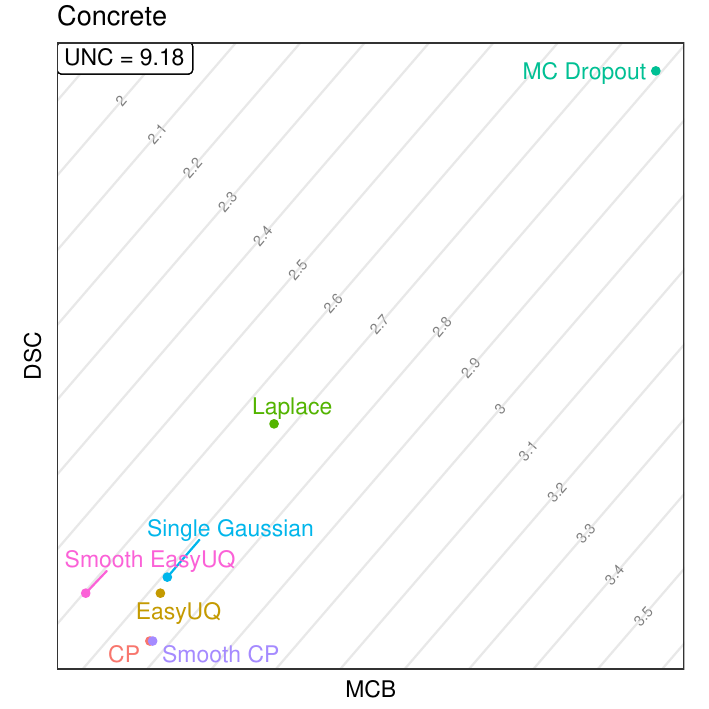}
\includegraphics[width = 0.245\textwidth]{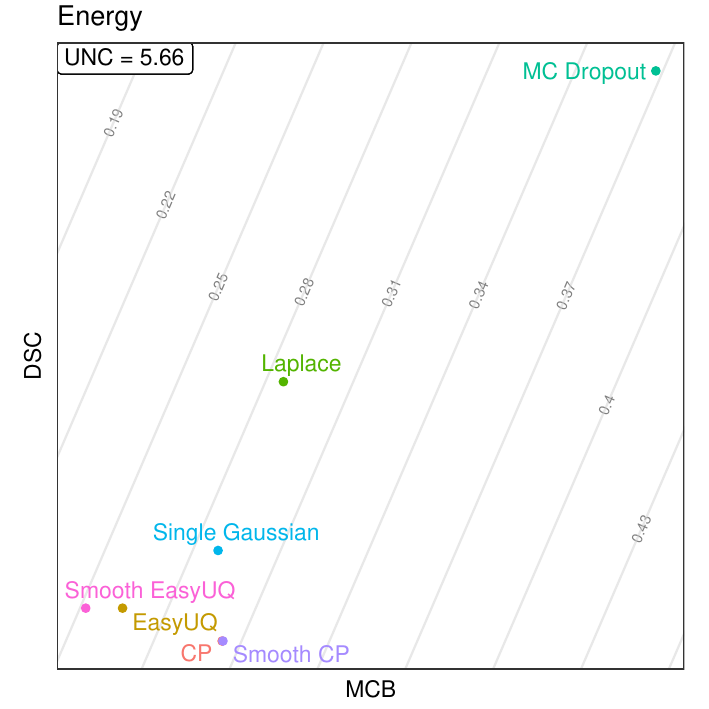}
\includegraphics[width = 0.245\textwidth]{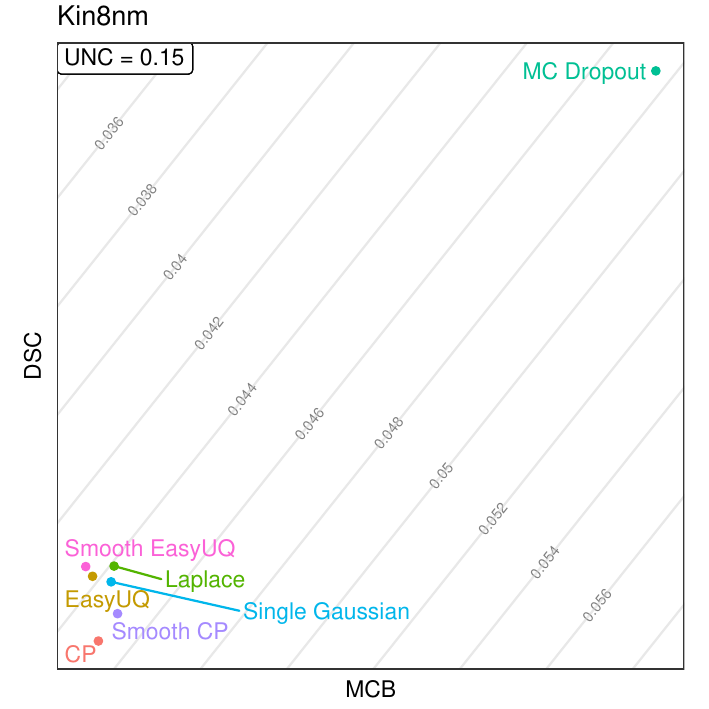}
\includegraphics[width = 0.245\textwidth]{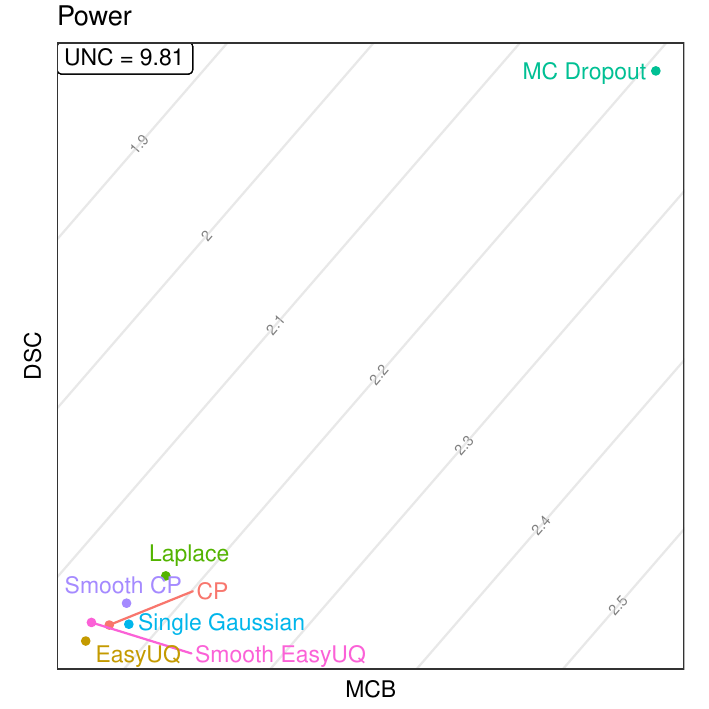}
\includegraphics[width = 0.245\textwidth]{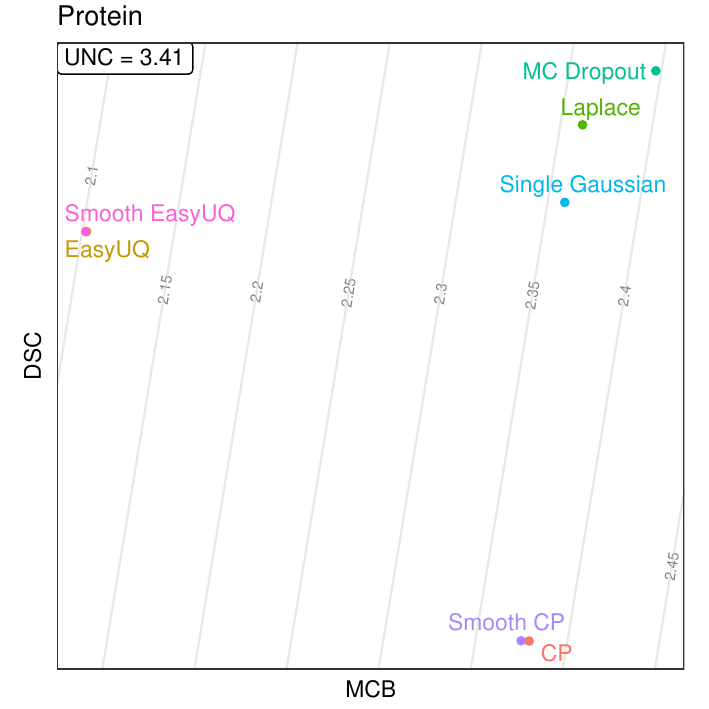}
\includegraphics[width = 0.245\textwidth]{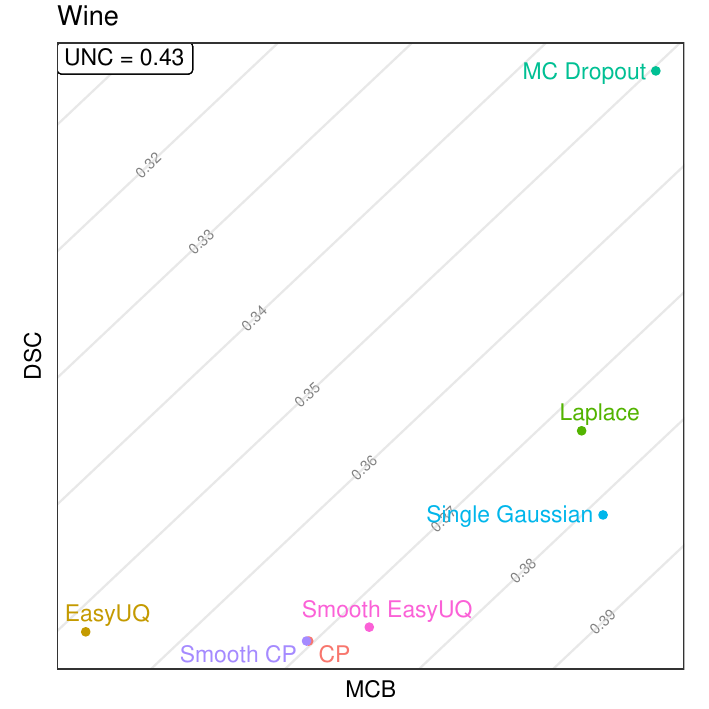} 
\includegraphics[width = 0.245\textwidth]{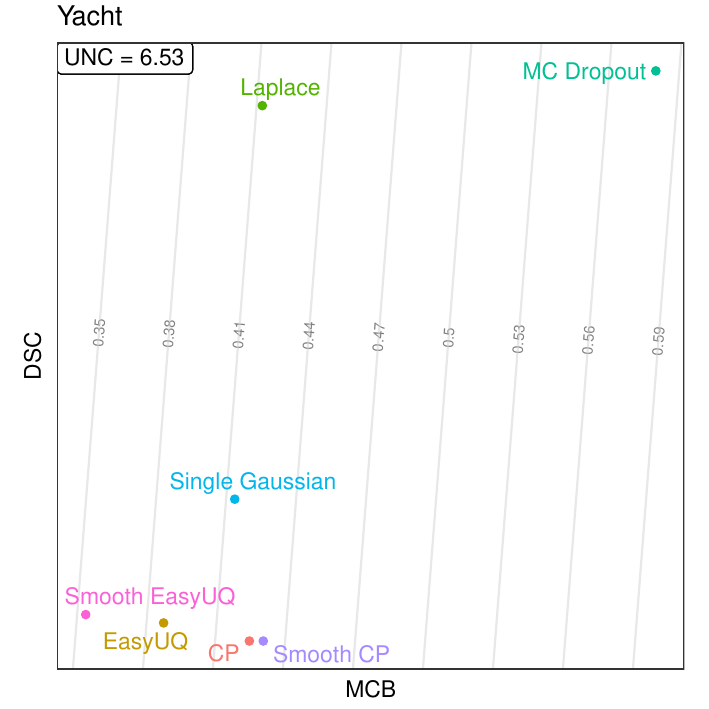}
\caption{$\MCBbar$--$\DSCbar$ plots for methods of uncertainty quantification for neural network based regression from the middle block in Table 6 of \citet{EasyUQ}.  The mean score $\Sbar$ is constant along the parallel lines.  \label{fig:ML}}
\end{sidewaysfigure}

The CP and EasyUQ distributions have finite support, and the Single Gaussian incurs normal distribution with a fixed variance, but varying mean.  For these three methods, we use the isotonicity-based decomposition of $\Sbar$ in the standard form from Section \ref{sec:ISO_standard}.  The Laplace method also employs normal distributions, but with varying mean and variances.  The MC Dropout technique yields mixtures of normal distributions, and the Smooth CP and Smooth EasyUQ distributions are mixtures of Student-$t$ distributions (or normal distributions as a limit case).  For these methods, we use the approximations described in Section \ref{sec:ISO_approximate}.  

The $\MCBbar$--$\DSCbar$ plots in Figure \ref{fig:ML} illustrate the mean score $\Sbar$ and the $\MCBbar$, $\DSCbar$, and $\UNCbar$ components for the eight datasets and seven methods, respectively.  The MC Dropout technique yields predictive distributions that are poorly calibrated, a finding that is well documented in the machine learning literature \citep{Gawlikowski2023}, though with high discrimination ability.  The predictive distributions generated by the Laplace method trade better calibration for diminished discrimination ability.  The simplistic Single Gaussian technique performs surprisingly well, typically with both the $\MCBbar$ and the $\DSCbar$ component being small relative to the competitors.  The EasyUQ and CP distributions generally are well calibrated, with low $\MCBbar$ components throughout, and often superior overall performance.  Smoothing of the discrete EasyUQ and CP distributions has only  small effects.  The only exception is for the EasyUQ forecast for the Wine dataset, which has only ten unique outcomes that correspond to quality levels, thus favoring the discrete basic EasyUQ distributions, which place all probability mass on this small set of outcomes.

\section{Discussion}  \label{sec:discussion} 

In line with the general idea of the CORP approach of \citet{Dimitriadis_2023} and \citet{Tilmann_Johannes_Calibration}, we have developed an isotonicity-based decomposition of the mean score $\Sbar$.  Both theoretically and computationally, the isotonicity-based decomposition serves as an attractive alternative to the Candille--Talagrand decomposition, which is of theoretical appeal, but yields degenerate decompositions in practice.  Remarkably, Proposition \ref{prop:inequalities} ensures that theoretical guarantees for the standard implementation from Section \ref{sec:ISO_standard} very nearly carry over to the approximate implementation described in Section \ref{sec:ISO_approximate}.   Code in \textsf{R} \citep{R} for the computation of the isotonicity-based decomposition and replication materials are available at \url{https://github.com/evwalz/isodisregSD} and \url{https://github.com/evwalz/paper_isocrpsdeco}, respectively.

Due to its linear computational complexity, the Hersbach decomposition is a viable option for decomposing $\Sbar$ for ensemble forecasts with a moderate number $m$ of members, even when the size $n$ of the evaluation set at \eqref{eq:data} is very large and the isotonicity-based approach with its quadratic complexity is not feasible.  We recommend that it be used in the modified form described in Section \ref{sec:HB}, which allows for extensions beyond the case of ensemble forecasts, as described in Appendix \ref{app:proofs1}.  A useful facet of the Hersbach decomposition is that it applies to general (nonnegatively) weighted sums (rather than simple averages only) of \crps\ scores \citep{Hersbach_2000}.  The isotonicity-based decomposition generalizes to weighted sums as well, as the theoretical guarantees for IDR \citep{IDR} continue to apply in weighted case, and software developed by Alexander Henzi (\url{https://github.com/AlexanderHenzi/isodistrreg}) handles the extension.  We leave details to future work.  

As noted, the desirable properties $(E_1)$, \dots, $(E_5)$ in the empirical case and $(P_1)$, \dots, $(P_5)$ in the population case remain valid for decomposition of the mean score under proper scoring rules other than the \crps.  For instance, in various applications a certain region of the potential range of the outcome is of particular interest, and predictive performance might then be assessed with emphasis on these regions.  In such settings, one may use versions of the \crps\ as proposed by \citet{Gneiting_Ranjan_2013}, namely, 
\begin{align*}
\crps_w(F,y) = \int_{-\infty}^\infty w(x) \, \bs(F(x), \one \{ y \leq x \}) \diff x
\end{align*}
and 
\begin{align*}
\crps_v(F,y) & = \int_0^1 v(\alpha) \, \qs(F^{-1}(\alpha),y) \diff \alpha, 
\end{align*}
where $w$ and $v$, respectively, are nonnegative weight functions.  In view of the universality property of IDR \citep[Theorem 2]{IDR}, the isotonicity-based decomposition extends naturally to means of these types of scores, while preserving its desirable properties. 

However, the isotonicity-based approach fails if a mean of logarithmic scores \citep[Section 4.1]{Gneiting2007a} is sought to be decomposed, for the logarithmic score, which allows for the comparison of density forecasts only, cannot be applied to the discrete IDR distributions.  While in principle isotonic recalibration by IDR, on which isotonicity-based decompositions are based, could be replaced by recalibration with other methods, it is not at all evident what type of technique ought to be used, and we are unaware of any such method that would share the optimality properties of IDR that underlie the theoretical guarantees enjoyed by the isotonicity-based approach.

Various authors have pondered the use of the \crps, which is favored by the meteorological and renewable energy literatures, as opposed to the logarithmic score, which is of particular popularity in econometrics and machine learning, with the choice arising both in the context of estimation via empirical score minimization and in the evaluation of predictive performance \citep{Gneiting2007a}.  For example, \citet[Appendix B]{DIsanto2018} argue that in neural network learning empirical score minimization in terms of the mean \crps\ is preferable to optimization of the logarithmic score.  In the evaluation of predictive performance, the availability of the theoretically supported and practically feasible isotonicity-based decomposition, in concert with the applicability of the score to discrete forecast distributions, strengthens arguments in favor of the \crps.  
\section*{Acknowledgements} 

We thank Tim Hewson, Kai Polsterer, and Johannes Resin for comments and discussion.  Tilmann Gneiting is grateful for  support by the Klaus Tschira Foundation.  The work of Eva-Maria Walz was funded by the German Research Foundation (DFG) through grant number 257899354.  Sebastian Arnold and Johanna Ziegel gratefully acknowledge financial support from the Swiss National Science Foundation.  Computations for the weather case study have been performed on UBELIX (\url{https://ubelix.unibe.ch/}), the HPC cluster of the University of Bern.

\bibliographystyle{abbrvnat}
\bibliography{biblio}

\appendix

\section{Technical details for the Brier score and quantile score based decompositions}  \label{app:BS_QS}

In this appendix we describe the Brier score (BS) and quantile score (QS) based decompositions from Sections \ref{sec:BS} and \ref{sec:QS} for the mean score $\Sbar$ of the forecast--observation pairs $(F_1,y_1), \dots, (F_n,y_n)$ at \eqref{eq:data}.  Both decompositions build on a general version of the pool-adjacent-violators (PAV) algorithm for nonparametric isotonic regression \citep{Ayer1955}.  While historically work on the PAV algorithm has focused on the mean functional \citep{4B, Robertson_1988, PAVA_in_R}, the algorithm yields optimal isotonic fits under any identifiable functional; see, e.g., \citet{Optimal} and \citet[Section 3.1]{Tilmann_Johannes_Calibration}. 

\subsection{Brier score based decomposition}  \label{app:BS}

For each threshold value $z \in \R$, we interpret $F_1(z), \dots, F_n(z)$ as probability forecasts for the binary event $\xi_i(z) = \one \{ y_i \leq z \}$, where $i = 1, \dots, n$.  We obtain calibrated forecasts $\cF_1(z), \dots, \cF_n(z)$ by applying the PAV algorithm for the mean functional on $\xi_1(z), \dots, \xi_n(z)$ with respect to the order induced by $F_1(z), \dots, F_n(z)$.  This yields the CORP decomposition of the mean Brier score 
\begin{align*}
\BSbar_{F(z)} = \frac{1}{n} \sum_{i=1}^n \bs \big( F_i(z), \xi_i(z) \big)
\end{align*}
as proposed by \cite{CORP}, namely,  
\begin{align*}
\BSbar_{F(z)} = \underbrace{\left( \BSbar_{F(z)} - \BSbar_{\cF(z)} \right)}_{\MCBbar_{\BS,z}} 
- \underbrace{\left( \BSbar_{\cF(z)} - \BSbar_{\hat{F}_\mg(z)} \right)}_{\DSCbar_{\BS,z}} 
+ \, \underbrace{\BSbar_{\hat{F}_\mg(z)}}_{\UNCbar_{\BS,z}},
\end{align*}
where $\hat{F}_\mg(z) = \frac{1}{n} \sum_{i=1}^n \xi_i(z)$ for $z \in \R$,
\begin{align*}
\BSbar_{\cF(z)} = \frac{1}{n} \sum_{i=1}^n \bs \big( \cF_i(z), \xi_i(z) \big)
\eqand
\BSbar_{\hat{F}_\mg(z)} = \frac{1}{n} \sum_{i=1}^n \bs \big( \hat{F}_\mg(z), \xi_i(z) \big).
\end{align*}
Integration of the $\MCBbar_{\BS,z}, \DSCbar_{\BS,z}$ and $\UNCbar_{\BS,z}$ components over $z\in \R$ yields the Brier score based score components and decomposition at \eqref{eq:BScomponents} and \eqref{eq:BSempirical}, respectively.  

Computationally, it suffices to run the PAV algorithm at $z \in \{ y_1, \dots, y_n \}$ and at the crossing points of the {\cdf}s $F_1, \dots, F_n$.  

\begin{proof}[Proof of Proposition \ref{prop:BS}]
We note that
\begin{align*}
\UNCbar_\BS & = \int \BSbar_{\hat{F}_\mg(z)} \diff z = \int \frac{1}{n} \sum_{i=1}^n \bs \big( \hat{F}_\mg(z), \xi_i(z) \big) \diff z \\
& = \frac{1}{n} \sum_{i=1}^n \int \big( \hat{F}_\mg(z) - \xi_i(z) \big)^2 \diff z = \frac{1}{n} \sum_{i=1}^n \crps(\hat{F}_\mg, y_i) = \UNCbar_0,
\end{align*}
which implies that ($E_5$) is satisfied.  Property ($E_1)$ is immediate.  \cite{CORP} show that $\MCBbar_{\BS,z}$ and $\DSCbar_{\BS,z}$ are nonnegative for all $z\in \R$ and thus ($E_2$) is satisfied.  Example \ref{app:3} implies that the decomposition is not degenerate, so ($E_3$) is satisfied.  Finally, suppose that $F_1 = \dots = F_n$.  Then for each $z \in \R$, the PAV algorithm for the mean functional on $\xi_1(z), \dots, \xi_n(z)$ with respect to the order induced by $F_1(z) = \dots = F_n(z)$ yields the constant calibrated forecast $\hat{F}_\mg(z)$.  Hence $\DSCbar_{\BS} = 0$, so that ($E_4)$ is satisfied.
\end{proof}

\begin{rem}  \label{rem:Fhat_i_need_not_be_cdfs}
The functions $\cF_1, \dots, \cF_n$ are not necessarily increasing and hence they generally fail to be {\cdf}s.  For instance, let $n = 2$ and $z < z'$.  If $F_1(z) < F_2(z)$, $F_1(z') = F_2(z')$ and $y_2 \leq z < z'< y_1$, then $\cF_2(z) = 1 > 1/2 = \cF_2(z')$, so $\cF_2$ is not increasing. 
\end{rem}

\subsection{Quantile score based decomposition}  \label{app:QS}

For each level $\alpha \in (0,1)$, we consider $F_1^{-1}(\alpha), \dots, F_n^{-1}(\alpha)$ as point forecasts in the form of the $\alpha$-quantile.  We apply the PAV algorithm for the $\alpha$-quantile functional on $y_1, \dots, y_n$ with respect to the order induced by $F^{-1}_1(\alpha), \dots, F^{-1}_n(\alpha)$ to yield calibrated $\alpha$-quantile forecasts $\cFinv^{-1}_1(\alpha), \dots, \cFinv^{-1}_n(\alpha)$.  This induces the CORP decomposition of the mean quantile score 
\begin{align*}
\QSbar_{F^{-1}(\alpha)} = \frac{1}{n} \sum_{i=1}^n \qs \big( F_i^{-1}(\alpha), y_i \big)
\end{align*}
as described by \citet[Section 3.3]{Tilmann_Johannes_Calibration} and \citet[Section 3.3]{Gneiting_2023}, namely,  
\begin{align*}
\QSbar_{F^{-1}(\alpha)}
= \underbrace{\big( \QSbar_{F^{-1}(\alpha)} - \QSbar_{\cFinv^{-1}(\alpha)}\big)}_{\MCBbar_{\QS,\alpha}} - \underbrace{\big( \QSbar_{\cFinv^{-1}(\alpha)} - \QSbar_{\hat{F}_\mg^{-1}(\alpha)} \big)}_{\DSCbar_{\QS,\alpha}} + \, \underbrace{\QSbar_{\hat{F}_\mg^{-1}(\alpha)}}_{\UNCbar_{\QS,\alpha}},
\end{align*}
where $\hat{F}_\mg^{-1}(\alpha)$ is the quantile function of the marginal empirical law of the outcomes $y_1, \dots, y_n$, 
\begin{align*}
\QSbar_{\cFinv^{-1}(\alpha)} = \frac{1}{n} \sum_{i=1}^n \qs \big( \cFinv^{-1}_i(\alpha), y_i \big),
\qquad
\QSbar_{\hat{F}_\mg^{-1}(\alpha)} = \frac{1}{n} \sum_{i=1}^n \qs \big( \hat{F}_\mg^{-1}(\alpha), y_i \big).  
\end{align*}
Integration of the $\MCBbar_{\QS,\alpha}, \DSCbar_{\QS,\alpha}$ and $\UNCbar_{\QS,\alpha}$ components over $\alpha\in (0,1)$ yields the quantile score based decomposition at \eqref{eq:QSempirical}.  

For an exact computation, the PAV algorithm needs to be run at all quantile levels $l/k$, where $k = 1, \dots, n$ and $l = 1, \dots, k - 1$, and at all crossing points of the quantile functions $F_1^{-1}, \dots, F_n^{-1}$.  In practice, it suffices to apply the PAV algorithm on a fine grid of quantile levels. 

\begin{proof}[Proof of Proposition \ref{prop:QS}]
In analogy to the proof of Proposition \ref{prop:BS}, we find that
\begin{align*}
\UNCbar_\QS & = \int_0^1 \QSbar_{\hat{F}_\mg^{-1}(\alpha)} \diff\alpha = \int_0^1 \frac{1}{n} \sum_{i=1}^n \qs \big( \hat{F}_\mg^{-1}(\alpha),y_i \big) \diff \alpha \\
& = \frac{1}{n} \sum_{i=1}^n \int_0^1 \qs \big( \hat{F}_\mg^{-1}(\alpha),y_i \big) \diff \alpha 
= \frac{1}{n} \sum_{i=1}^n \crps(\hat{F}_\mg, y_i) = \UNCbar_0,
\end{align*}
and hence ($E_5$) is satisfied.  Property ($E_1$) is clear by definition.  Theorem 3.3 of \citet{Tilmann_Johannes_Calibration} implies that $\MCBbar_{\QS,\alpha}$ and $\DSCbar_{\QS,\alpha}$ are nonnegative for all $\alpha \in (0,1)$ and thus ($E_2$) is satisfied.  Example \ref{app:3} shows that the decomposition is not degenerate, i.e., ($E_3$) is satisfied.  Finally, suppose that $F_1 = \dots = F_n$.  Then for each $\alpha \in (0,1)$, applying the PAV algorithm on $y_1, \dots, y_n$ with respect to the order induced by $F^{-1}_1(\alpha) = \dots = F^{-1}_n(\alpha)$ yields the constant calibrated forecast $\cFinv^{-1}(\alpha) = \hat{F}_\mg^{-1}(\alpha)$ and hence $\DSCbar_\QS = 0$, i.e., ($E_4)$ is satisfied.
\end{proof}

\begin{rem}  \label{rem:Fhat_i_need_not_be_cdfs_quantiles}
In analogy to the statements in Remark \ref{rem:Fhat_i_need_not_be_cdfs}, the functions $\cFinv^{-1}_1, \dots, \cFinv^{-1}_n$ are not necessarily increasing and hence may not be quantile functions.  For example, let $n = 2$ and $\alpha < \alpha' < 1/2$, and suppose that $y_1 < y_2$, $F_1^{-1}(\alpha) < F_2^{-1}(\alpha)$, and $F_1^{-1}(\alpha') = F_2^{-1}(\alpha')$.  Then $\cFinv^{-1}_2(\alpha) = y_2 > y_1 = \cFinv^{-1}_2(\alpha')$ whence $\cFinv^{-1}_2$ is not increasing. 
\end{rem}

\section{Technical details for the original and modified Hersbach decompositions}  \label{app:HB}

As in Section \ref{sec:HB}, we consider a collection of the form at \eqref{eq:data} of forecast--outcome pairs $(F_1,y_1), \dots, (F_n,y_n)$, where for $i = 1, \dots, n$, the forecast $F_i$ is the empirical \cdf\ of a fixed number $m$ of numbers $x_1^i \le  \dots \le x_m^i$.  \cite{Hersbach_2000} implicitly assumes that $y_i \notin \{ x_1^i, \dots, x_m^i \}$ for $i = 1, \dots, n$.  If this condition is not satisfied, the extension of the original Hersbach decomposition at \eqref{eq:Hersbach}, which is implemented in the \textsf{R} function \texttt{crpsDecomposition} from the verification package (\url{https://rdrr.io/cran/verification/}), is problematic.  Our suggested modified Hersbach decomposition at \eqref{eq:Hersbach_mod} resolves this issue, as illustrated graphically in Figure \ref{fig:HB}.

\begin{figure}[t]
\caption{Adaptation of Figure 2 from \cite{Hersbach_2000} with the empirical \cdf\ of $x_1 < \dots < x_5$ and outcome $y$.  \citet{Hersbach_2000} assumes that $y \notin \{ x_1, \dots, x_5 \}$ and divides the quantity $x_{\ell+1} - x_\ell$ for $\ell = 1, \dots, m-1$ into $\alpha_\ell$ and $\beta_\ell$, as illustrated in the left panel.  When $y = x_3$ the 
original decomposition sets $\alpha_2 = \beta_3 = 0$.  However, according to display (26) in \cite{Hersbach_2000}, if $y \uparrow x_3$ then $\alpha_2 \to x_3 - x_2$, $\beta_2 \to 0$, and $\beta_3 = x_4 - x_3$, and if $y \downarrow x_3$ then $\alpha_2 = x_3 - x_2$, $\alpha_3 \to 0$, and $\beta_3 \to x_4 - x_3$.  This suggests that $\alpha_2 = x_3 - x_2$, $\alpha_3 = 0$, $\beta_2 = 0$, and $\beta_3 = x_4 - x_3$ when $y = x_3$, as indicated in the right panel and in accordance with the quantity $\bar{f}_3$ in the modified Hersbach decomposition.   
\label{fig:HB}}
\bigskip
\includegraphics[width = \textwidth]{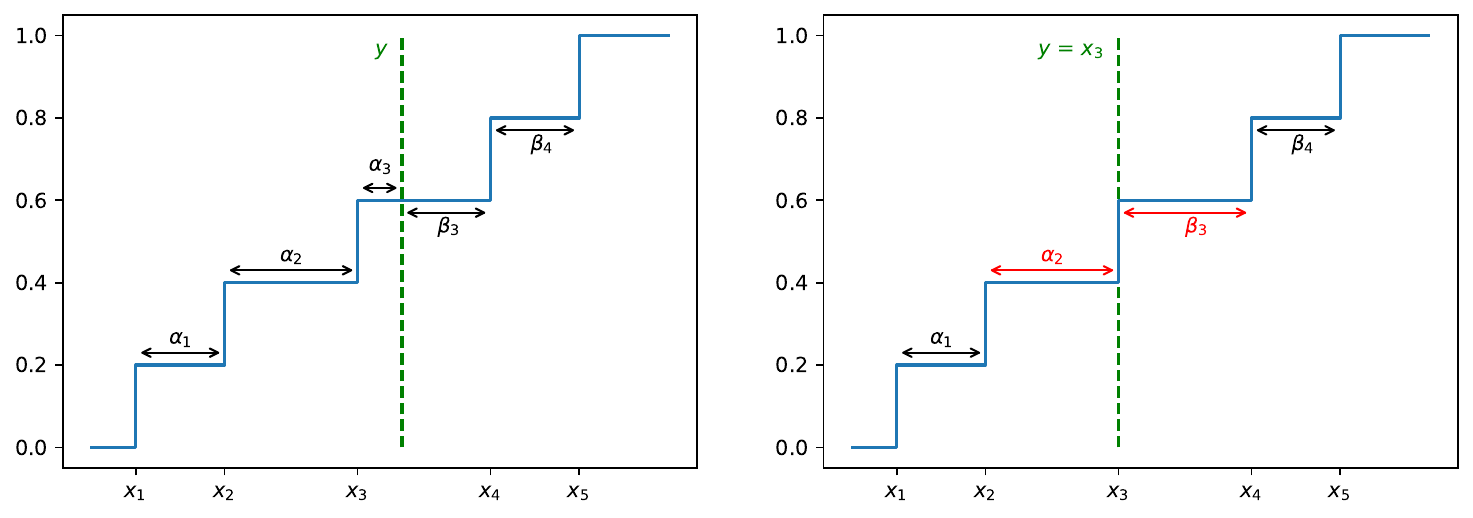}
\end{figure}

We proceed to a comparison of the orginal with the modified Hersbach decomposition.  For $i = 1, \dots, n$, \cite{Hersbach_2000} defines the quantities
\begin{align*}
\alpha_\ell^i &
= (x^i_{\ell+1} - x^i_\ell) \, \one \{ y_i > x_{\ell+1} \} + (y_i-x_\ell) \, \one \{ x^i_\ell < y_i < x^i_{\ell+1} \}, \\
\beta^i_\ell & = (x^i_{\ell+1} - x^i_\ell) \, \one \{ y_i < x^i_{\ell} \} + (x^i_{\ell+1} - y_i) \, \one \{ x^i_\ell < y_i < x^i_{\ell+1} \},
\end{align*}
for $\ell = 1, \dots, m - 1$, and 
\begin{align*}
\alpha_m^i = (y_i-x_{m}^i) \, \one \{ y_i > x_m^i \} \eqand \beta_0^i = (x_1^i - y_i) \, \one \{ y_i < x_1^i \}. 
\end{align*}
For $\ell = 1, \dots, m - 1$, let $\bar{\alpha}_\ell = (1/n) \sum_{i=1}^n \alpha_\ell^i$, $\bar{\beta}_\ell = (1/n) \sum_{i=1}^n \beta_\ell^i$, $\bar{g}_\ell = \bar{\alpha}_\ell + \bar{\beta}_\ell$, and $\bar{o}_\ell = \bar{\beta}_\ell / \bar{g}_\ell$.  To complete the specification, let $\bar{o}_0 = (1/n) \sum_{i=1}^n \one \{ y_i < x_1^i \}$, $\bar{g}_0 = \one \{ \bar{o}_0 \neq 0 \} \bar{\beta}_0 /\bar{o}_0$, $\bar{o}_m = (1/n) \sum_{i=1}^n \one \{x_m^i < y_i\}$, and $\bar{g}_m = \one \{ \bar{o}_m \neq 0 \} \bar{\alpha}_m / (1 - \bar{o}_m)$, where $\bar{\beta}_0 = (1/n) \sum_{i=1}^n \beta_0^i$ and ${\alpha}_m = (1/n) \sum_{i=1}^n \alpha_m^i$. 

As before, let $p_\ell = \ell/m$ for $\ell = 0, \dots, m$.  \cite{Hersbach_2000} defines the miscalibration component as
\begin{align*}
\MCBbar_\HBO = \sum_{\ell=0}^m \bar{g}_\ell \left( p_\ell - \bar{o}_\ell \right)^2.
\end{align*}
In contrast, we let
\begin{align*}
\MCBbar_{\HB} = \sum_{\ell=1}^{m-1} \bar{g}_\ell \left(p_\ell - \bar{f}_\ell \right)^2,
\end{align*}
where $\bar{f}_\ell = (1/n) \sum_{i=1}^n \bar{f}^i_\ell $ with $\bar{f}^i_\ell = (1/{\bar{g}_\ell}) \, \one \{ y_i < x_{\ell+1}^i \, \} (\alpha_\ell^i + \beta_\ell^i)$ for $i = 1, \dots, n$ and $\ell = 1, \dots, m - 1$.  In other words, \cite{Hersbach_2000} includes terms for $l = 0$ and $l = m$ in the miscalibration component and compares the nominal level $p_\ell$ with the quantity $\bar{o}_\ell$, which approximates the frequency of an outcome below the midpoint of bin $l$.  In contrast, we omit the outer terms and compare $p_\ell$ with $\bar{f}_\ell$, which approximates the frequency of an outcome below the right endpoint of bin $l$. 

\begin{proof}[Proof of Proposition \ref{prop:HB_decomp}]
By definition, both decompositions are exact and the uncertainty component $\UNCbar_0$ depends only on the outcomes, i.e., ($E_1$) and ($E_5$) are satisfied.  Example \ref{app:3} shows that ($E_3$) is satisfied, and that ($E_2$) fails to hold for the modified Hersbach decomposition.  Consider the sample $(F,y_1), (F,y_2)$ with $F = (\delta_{-1/2} + \delta_{1/2})/2$, $y_1 = - 1/6$ and $y_2 = 1/6$.  Then $\Sbar = 1/4$ and $\UNCbar_0 = 1/12$.  Moreover, $\bar{g}_1 = 1$, $\bar{g}_0 = \bar{g}_2 = 0$, $\bar{o}_1 = 1/2$, $\bar{o}_0 = \bar{o}_2 = 0$, and $\bar{f}_1 = 1$.  Thus $\MCBbar_\HBO = 0$, $\MCBbar_{\HB} = 1/4$, $\DSCbar_\HBO = - 1/6$, and $\DSCbar_\HB = 1/12$.  This demonstrates that the original Hersbach decomposition does not satisfy ($E_2$) and ($E_4$) and that ($E_4$) fails to hold for the modified decomposition as well.  Numerical examples in \cite{Hersbach_2000} show that ($E_3$) is satisfied for the original Hersbach decomposition. 
\end{proof}

\section{Relaxations of the stochastic order}  \label{app:ISO}

Consider any partial order $\leq'$ on $\PR$, which is weaker than the stochastic order in the sense that $G \st H$ implies $G \leq' H$ for $G, H \in \PR$.  Possible choices include the almost-first-stochastic-dominance order proposed by \cite{Leshno_Levy_2002} or stochastic dominance of order $(1 + \gamma)$ as proposed by \cite{Muller}.  If there are only few forecasts in a sample $(F_1,y_1), \dots, (F_n,y_n) \in \PR \times \R$ that are comparable with respect to $\st$, one could think of applying IDR with respect to $\leq'$ instead of $\st$ in order to obtain more comparable forecasts.  In this appendix, we explain why such an approach is bound to fail.

Let $Y$ be a random variable and $F$ be a random forecast defined on the same probability space.  Recall from Section \ref{sec:isotonic} that ICL forms the population version of IDR \citep[Proposition 4.1]{ICL}.  In analogy to Definition 3.1 of \cite{ICL}, one could define the $\s$-lattice generated by $F$ with respect to the weaker order $\leq'$ as $\LL'(F) = \{F^{-1}(B) \mid B \in \BPR \cap \mathcal{U}' \}$, where $\mathcal{U}'$ denotes the family of all upper sets in $\PR$ with respect to $\leq'$.  However, if the space $\PR$ equipped with the partial order $\leq'$ and the topology of weak convergence satisfies Assumption C.1 of \cite{ICL}, the corresponding notion of isotonic calibration, namely, $P_{Y \mid \LL'(F)} = F$, fails to be intuitive for two reasons.  First, auto-calibration does not imply the respective notion of calibration.  Second, $G \leq' H$ already implies $G \st H$ for all $G$ and $H$ in the support of $F$ by Theorem 3.3 of \cite{ICL}. Clearly,  this implication may only hold if $\leq' $ equals $\st$ on the support of $F$, which is violated for any $\leq'$ that is strictly weaker than $\st$, contrary to the scope of a relaxation.  Moreover, there is no theoretical guarantee that the corresponding miscalibration term $\MCB_{\ISO'} = \E \, \crps(F,Y) - \E \, \crps(P_{Y\mid \LL'(F)},Y)$ is nonnegative. 

\section{Proofs for Section \ref{sec:population}}  \label{app:proofs}

\subsection{Proofs for Section \ref{sec:decomposition} and extensions}   \label{app:proofs1}

\begin{proof}[Proof of Proposition \ref{prop:HB_general}]
Following Appendix A in \cite{Candille_2005}, we apply the change of variable $z \mapsto p = F(z)$ to demonstrate that $\E \, \crps(F,Y)$ can be represented as
\begin{align*}
\E \int_S (F(z) - \one \{ F(Y) \leq F(z) \})^2 \diff z + \E \int_S (2 F(z) - 1)(\one \{ F(Y) \leq F(z) \} - \one \{ Y \leq z \} ) \diff z,
\end{align*}
where $S = \{ z \in \R \mid (F(z) - \one \{ Y \leq z \})^2 > 0 \}$.  The indicator is essential, since if $F(Y) = 0$ then $\one \{ F(Y) \leq F(z) \} = 1$ and the integrals may not exist.  We decompose $S$ into the disjoint sets $S_1 = S \cap \{z \in \R \mid F(z) > 0 \}$ and $S_2 = S \cap \{z \in \R \mid F(z) = 0 \} = \{ z \in \R \mid Y \leq z, F(z) = 0 \}$, and use the equivalence $\one \{ F(Y) \leq F(z) \} - \one \{Y \leq z \} = \one \{ Y > z, F(Y) = F(z) \}$ to show that
\begin{align*}
\E \, \crps(F,Y)  
& =  \E \int_{S_1} (F(z) - \one \{ F(Y) \leq F(z) \})^2 \diff z + \E \int_{S_2} \one \{ Y \leq z, F(z) = 0 \} \diff z \\
& \quad + \E \int_S (2F(Y) - 1) \, \one \{ Y > z, F(Y) = F(z) \} \diff z \\
& =  \E \int_{S_1} (F(z) - \one \{ F(Y) \leq F(z) \})^2 \diff z + \MS,
\end{align*}
where $\MS$ is given at \eqref{eq:MS_Hersbach}.

We have $\tau(A) \leq \E \int_A 1 \diff \nu_F(u) = \E (\nu_F(A)) = \mu (A)$ for $A \in \BU$, i.e., $\tau$ is absolutely continuous with respect to $\mu$.  Hence $\tau$ has a density $f$ with respect to $\mu$, and we find that
\begin{align*}
\E \, \crps (F,Y) 
& = \E \int_S (F(z) - \one \{ F(Y) \leq F(z) \})^2 \diff z + \MS \nonumber \\
& = \E \int_0^1 (p - \one \{ F(Y) \leq p \})^2 \diff \nu_F(p) + \MS \nonumber \\
& = \int_0^1 p^2 \diff \mu(p) - \int_0^1 (2p-1) \, \diff \tau(p) + \MS \nonumber \\
& = \int_0^1 p^2 \diff \mu(p) - \int_0^1 (2p-1) \, f(p) \diff \mu(p) + \MS\nonumber \\
& = \int_0^1 (p-f(p))^2 \diff \mu(p) + \int_0^1 f(p) \, (1-f(p)) \diff \mu(p) + \MS,
\end{align*}
which yields the claimed decomposition.
\end{proof}

In the following corollary to Proposition \ref{prop:HB_general}, which is a more general result than Corollary \ref{cor:HB_2}, we consider forecast--observation pairs $(F_1,y_1), \dots, (F_n,y_n)$, where for each $i = 1, \dots, n$, $F_i$ is a distribution with a finite number $m_i$ of support points $x^i_1 < \dots < x^i_{m_i}$ and (cumulative) probability values $p^i_1 < \dots < p^i_{m_i}$, so that $F_i(x^i_\ell) = p^i_\ell$ for $\ell = 1, \dots, m_i$.  Let $0 < \hat{p}_1 < \ldots < \hat{p}_M = 1$ be the unique probability values from the set $\{ p^i_\ell \mid i = 1, \dots, n; \, \ell = 1, \dots, m_i \}$.  For $i = 1, \dots, n$ and $j = 1, \dots, M - 1$, we define
\begin{align*}
\sigma^i_j = \begin{cases}
\ell & \textrm{if} \quad \hat{p}_j = p^i_\ell, \\
0    & \textrm{if} \quad \hat{p}_j \notin \{ p_1^i, \dots, p_{m_i}^i \}.
\end{cases}
\end{align*}

\begin{cor}  \label{cor:HB_for_step_functions}
Assume that\/ $\prob$ is the empirical measure of forecast--observation pairs\/ $(F_1,y_1), \dots, (F_n,y_n)$, where each\/ $F_i$ is a distribution with finite support as described above.  Then 
\begin{equation}  \label{eq:MCB_HB_finite_support}
\MCB_\HB = \sum_{j=1}^{M-1} \hat{g}_j (\hat{p}_j - \hat{f}_j)^2
\end{equation}
where, for $j = 1, \dots, M-1$,
\begin{align} 
\hat{g}_j & = \frac{1}{n} \sum_{i=1}^n \one \{ \sigma_j^i \neq 0 \} \left( x_{\sigma_j^i+1}^i - x^i_{\sigma_j^i} \right) \! , \label{eq:g_j} \\
\hat{f}_j & = \frac{1}{n \hat{g}_j} \sum_{i=1}^n \one \{ F_i(y_i) \leq \hat{p}_j \} \one \{ \sigma_j^i \neq 0 \} \left( x_{\sigma_j^i+1}^i - x^i_{\sigma_j^i} \right) \! .  \label{eq:f_j}
\end{align}
\end{cor}

\begin{proof}
For $i = 1, \dots, n$, let $\nu_i$ be the image measure of $F_i$ with respect to the Lebesgue measure, i.e., 
\begin{align*}
\nu_i = \sum_{j=1}^{M-1} \delta_{\hat{p}_j} \one \{ \sigma_j^i \neq 0 \} \left( x_{\sigma_j^i+1}^i - x^i_{\sigma_j^i} \right) \! ,
\end{align*}
and thus, $\mu = \sum_{j=1}^{M-1} \delta_{\hat{p}_j} \hat{g}_j$, where $\hat{g}_j$ is given at \eqref{eq:g_j}.  Therefore, for any $A\in \mathcal{B}(0,1)$, we have
\begin{align*}
\tau(A) 
& = \E \int_A \one \{ F(Y) \leq u \} \diff \nu_F (u) \\
& = \frac{1}{n} \sum_{i=1}^n \sum_{j=1}^{M-1} \delta_{\hat{p}_j}(A) \one \{ F_i(y_i) \leq \hat{p}_j \} \one \{ \sigma_j^i \neq 0 \} \left( x_{\sigma_j^i+1}^i - x^i_{\sigma_j^i} \right) \\
& = \sum_{j=1}^{M-1} \delta_{\hat{p}_j}(A) \frac{1}{n} \sum_{i=1}^n \one \{ F_i(y_i) \leq \hat{p}_j \} \one \{ \sigma_j^i \neq 0 \} \left( x_{\sigma_j^i+1}^i - x^i_{\sigma_j^i} \right) = \sum_{j=1}^{M-1} \delta_{\hat{p}_j}(A) \hat{f}_j \hat{g}_j.
\end{align*}
We conclude that the Radon--Nikodym derivative of $\tau$ with respect to $\mu$ is $f(\hat{p}_j) = \hat{f}_j$ for $j = 1, \dots, M-1$, where $\hat{f}_j$ is given at \eqref{eq:f_j}.
\end{proof}

To specialize Corollary \ref{cor:HB_for_step_functions} to the ensemble setting of Corollary \ref{cor:HB_2}, let $m_i = m$ and $p_\ell^i = \ell/m$ for $i = 1, \dots, n$ and $\ell = 1, \dots, m - 1$.  Then $M = m$, $\hat{p}_j = j/m$, and the quantities in \eqref{eq:g_ell} and \eqref{eq:g_j} coincide, as do the first quantity in \eqref{eq:f_ell} and that in \eqref{eq:f_j}.  

\begin{proof}[Proof of Corollary \ref{cor:HB_1}]
Since $F^{-1}$ is almost surely absolutely continuous, for any $0 < a < b < 1$, we have almost surely
\begin{align*}
\nu_F([a,b)) = \lambda (F^{-1}([a,b))) = F^{-1}(b) - F^{-1}(a) = \int_{F^{-1}(a)}^{F^{-1}(b)} \diff p = \int_a^b \frac{\diff}{\diff p} F^{-1}(p) \diff p.
\end{align*}
That is, the random measure $\nu_F$ almost surely possesses a density $(\diff/\diff p) \, F^{-1}(p)$ with respect to the Lebesgue measure, and it follows that the measure $\mu$ has density $\gamma$ at \eqref{eq:gamma} with respect to the Lebesgue measure.  Since for $A \in \BU$, 
\begin{align*}
\tau(A) & = \E \int_A \one \{F(Y) \leq p \} \diff \nu_F(p) 
          = \int_A \E \left( \one \{ F(Y) \leq p \} \, \frac{\diff}{\diff p} F^{-1}(p) \right) \diff p,
\end{align*}
the density $f$ of the measure $\tau$ with respect to $\mu$ is given as stated at \eqref{eq:f}. 
\end{proof}

The following example relates to the case study on probabilistic quantitative precipitation forecasts in Section \ref{sec:weather}, where it applies to the BMA, EMOS, and HCLR forecasts, respectively.

\begin{example}  \label{ex:prec_forecasts} 
Let $(F_1,y_1), \dots, (F_n,y_n)$ be forecast--observation pairs for a nonnegative (possibly, censored) quantity, so that $y_i \geq 0$ for $i = 1, \dots, n$.  Suppose that, for $i=1, \dots, n$, 
\begin{align*}
F_i(x) = \begin{cases}
0                               & \textrm{for} \quad x < 0, \\
p_0^i + \int_0^x f_i(t) \diff t & \textrm{for} \quad x \ge 0,
\end{cases}
\end{align*}
for some $0 \leq p_0^i < 1$ and a strictly positive continuous function $f_i: (0,\infty) \to \R_+$ with $\int_0^\infty f_i(t) \diff t = 1 - p_0^i$.  Then $F_i^{-1}$ is absolutely continuous and has derivative $f_i(F_i^{-1}(p))^{-1}$ for $p \in (p_0^i,1)$ and zero otherwise.  Hence, $\MCBbar_\HB = \int_0^1 (p-f(p))^2 \, \gamma (p) \diff p$ by Corollary \ref{cor:HB_1}, where
\begin{align*}
\gamma (p) = \frac{1}{n} \sum_{i=1}^n \frac{1}{f_i(F_i^{-1}(p))} \one_{(p_0^i,1)}(p), \;
f(p) = \frac{1}{n\gamma (p)} \sum_{i=1}^n \one \{ F_i(y_i) \leq p \} \frac{1}{f_i(F_i^{-1}(p))} \one_{(p_0^i,1)}(p)
\end{align*}
for $p \in (0,1)$ with $\gamma(p) > 0$, and $f(p) = 0$ otherwise.
\end{example}

\subsection{Proofs for Section \ref{sec:properties}}  \label{app:proofs2}

\begin{proof}[Proof of Theorem \ref{thm:calibration}]  
Concerning part (a), we consider the Brier score based decomposition of $\Sbar$ and apply Fubini's theorem to obtain  
\begin{align}
\MCB_\CT & = & \int \left( \E \big( F(z) - \one \{ Y \leq z \} \big)^2 - \E \big( \prob( Y \leq z \mid F) - \one \{ Y \leq z \} \big)^2 \right) \diff z,  \label{eq1_proof_thm} \\
\DSC_\CT & = & \int \left( \E \big(F_\mg(z) - \one \{Y \leq z \} \big)^2 - \E \big( \prob( Y \leq z \mid F) - \one \{ Y \leq z \} \big)^2 \right) \diff z.  \label{eq2_proof_thm}
\end{align}
Recall that for any $z \in \R$, the expectation $\E \, ( \one \{ Y \leq z \} - p)^2$ is minimized by $\prob( Y \leq z \mid F)$ over all $\s(F)$-measurable random variables $p$, and this minimizer is $\prob$-almost surely unique.  Since $F(z)$ and the constant $F_\mg(z)$ are $\s(F)$-measurable, it follows from \eqref{eq1_proof_thm} and \eqref{eq2_proof_thm} that $\MCB_\CT \geq 0$ and $\DSC_\CT \geq 0$, respectively.  Equality in \eqref{eq1_proof_thm} holds if, and only if, $F$ is auto-calibrated.  Equality in \eqref{eq2_proof_thm} holds if, and only if, $P_{Y \mid F} = F_\mg$, i.e., $\prob( Y \leq z \mid F) = F_\mg(z)$ for all $z \in \R$.

For part (b), in analogy to the above, we find that 
\begin{align}
\MCB_\ISO & = \int \! \left( \E \big( \bar{F}(z) - \one \{ Y > z \} \big)^2 - \E \, \big( \prob( Y > z \mid \LL(F)) - \one \{ Y > z \} \big)^2 \right) \! \diff z, \label{eq:iso1} \\
\DSC_\ISO & = \int \! \left( \E \big( \bar{F}_\mg(z) - \one \{ Y > z \} \big)^2 - \E \big( \prob( Y > z \mid \LL(F)) - \one \{ Y > z \} \big)^2 \right) \! \diff z, \label{eq:iso2}
\end{align}
where $\bar{F}(z) = 1 - F(z)$, and $\bar{F}_\mg(z) = 1-F_\mg(z)$.
Recall that for any $z\in \R$, the expectation $\E( \one \{ Y > z \} - p)^2$ is minimized by $\prob( Y > z \mid \LL(F))$ over all $\LL(F)$-measurable random variables $p$, and the minimizer is $\prob$-almost surely unique.  Since $\bar{F}(z)$ and the constant $\bar{F}_\mg(z)$ are $\LL(F)$-measurable, it follows directly that $\MCB_\ISO \geq 0$ and $\DSC_\ISO \geq 0$. Equality in \eqref{eq:iso1} holds if, and only if, $F$ is isotonically calibrated, and equality in \eqref{eq:iso2} holds if, and only if,  $P_{Y \mid \LL(F)} = F_\mg$.  

To demonstrate part (c), it suffices to observe from \citet[Lemma 5.4]{ICL} that threshold calibration is equivalent to $\prob(Y \leq z \mid \LL(F(z))) = F(z)$ for $z\in \R$.  The rest of the argument is analogous to the above. 

Finally, for part (d), recall that for $\alpha \in (0,1)$, a random variable is a conditional quantile $q_\alpha(Y \mid \LL(F^{-1}(\alpha)))$ if, and only if, it minimizes $\E \, \QS_\alpha(X, Y)$ over all $\LL(F^{-1}(\alpha))-$measurable random variables $X$, see \cite{ICL}.  It follows that $\MCB_\QS \geq 0$ and $\DSC_\QS \geq 0$.  Assume that $F$ is quantile calibrated; then $q_\alpha \big( Y \mid \LL \big( F^{-1}(\alpha) \big) \big) = F^{-1}(\alpha)$ for $\alpha \in (0,1)$ and hence $\MCB_\QS = 0$.  Conversely, if $\MCB_\QS = 0$ then Fubini's theorem implies
\begin{align*}
\int_0^1 \left( \E \, \qs \big( F^{-1}(\alpha), Y \big) - \E \, \qs \big(q_\alpha \big( Y \mid \LL(F^{-1}(\alpha)) \big), Y \big) \right) \! \diff \alpha = 0.
\end{align*}
Since the integrand is non-negative, it follows that $q_\alpha \big( Y \mid \LL(F^{-1}(\alpha)) \big) = F^{-1}(\alpha)$ for almost all $\alpha \in (0,1)$ and, hence, there exists a Lebesgue null set $N \subseteq (0,1)$ with $q_\alpha(Y \mid \LL(F^{-1}(\alpha))) = F^{-1}(\alpha)$ for all $\alpha \in (0,1) \setminus N$.  Assume for a contradiction that $N \neq \emptyset$ and consider $\alpha_0 \in N$.  Choose $(\alpha_n)_{n \in \N} \subseteq (0,1) \setminus N$ with $\alpha_n \uparrow \alpha_0$ as $n \to \infty$.  Since $F^{-1}(\alpha_n) \to F^{-1}(\alpha_0)$ almost surely and $\textrm{qs}_{\alpha_n}(\cdot,y) \to \textrm{qs}_{\alpha_0}(\cdot, y)$ pointwise for any $y\in \R$, it follows that $\textrm{qs}_{\alpha_n}(F^{-1}(\alpha_n), Y) \to \textrm{qs}_{\alpha_0}(F^{-1}(\alpha_0), Y)$ almost surely, and hence, $\E \, \textrm{qs}_{\alpha_n}(F^{-1}(\alpha_n), Y) \to \E \, \textrm{qs}_{\alpha_0}(F^{-1}(\alpha_0), Y)$ by dominated convergence.   Analogously, $\E \, \textrm{qs}_{\alpha_n}(X,Y) \to \E \, \QS_{\alpha_0}(X,Y)$ for $X = q_{\alpha_0}(Y \mid \LL(F^{-1}(\alpha_0)))$ and, hence, $\E \, \textrm{qs}_{\alpha_0}(X,Y) \ge  \E \, \textrm{qs}_{\alpha_0}(F^{-1}(\alpha_0), Y)$ since $\E \, \textrm{qs}_{\alpha_n}(X, Y) \ge \E \, \textrm{qs}_{\alpha_n}(F^{-1}(\alpha_n), Y)$ for all $n\in \N$.  This shows that $q_\alpha \big( Y \mid \LL(F^{-1}(\alpha)) \big)$ is an $\alpha$-quantile of $F$ for $\alpha \in (0,1)$.  By construction in Section 6 of \citet{ICL}, $q_\alpha \big( Y \mid \LL(F^{-1}(\alpha)) \big)$ is the smallest possible minimizer of $\E \, \qs(X, Y)$, so it coincides with $F^{-1}(\alpha)$ for all $\alpha \in (0,1)$ and, hence, $N = \emptyset$.  Clearly, $\DSC_\QS = 0$ if $q_\alpha \big( Y \mid \LL(F^{-1}(\alpha)) \big) = q_\alpha(Y)$ for $\alpha\in (0,1)$.  Conversely, if $\DSC_\QS = 0$ then $q_\alpha \big( Y \mid \LL(F^{-1}(\alpha)) \big) = q_\alpha(Y)$ for $\alpha\in (0,1)$.
\end{proof}

\begin{proof}[Proof of Corollary \ref{cor:order_of_MCB}]  
For any $z \in \R$, $P_{Y \mid F}(\cdot, (z,\infty))$ minimizes $\E(p - \one \{ Y > z \})^2$ over all $\s(F)$-measurable random variables $p$, and hence, also over all $\LL(F)$-measurable random variables since any $\LL(F)$-measurable random variable is also $\s(F)$-measurable, see \citet[Lemma 3.1]{ICL}.  Thus, we apply Fubini to derive 
\begin{align*}
\E \, \crps(P_{Y \mid F},Y) 
& = \int \E \, (P_{Y \mid F}(\cdot, (z,\infty)) - \one \{ Y > z \} )^2 \diff z \\
& \leq \int \E \, (P_{Y \mid \LL(F)}(\cdot, (z,\infty)) - \one \{ Y > z \} )^2 \diff z = \E \, \crps(P_{Y\mid \LL(F)},Y),
\end{align*}
which implies $\MCB_\CT \geq \MCB_\ISO$.  Moreover, for any $z \in \R$ we know that $\LL(F(z)) \subseteq \overline{\LL(F)}$, where for any $\s$-lattice $\A\subseteq \F$, $\bar{\A}$ denotes the $\s$-lattice which consists of all complements of elements in $\A$.  Hence, we may argue similarly that 
\begin{align*}
\E \, \crps(P_{Y \mid \LL(F)},Y) 
& = \int \E (1-P_{Y \mid \LL(F)}(\cdot, (z,\infty)) - \one \{ Y \leq z\})^2 \diff z \\ 
& \leq \int \E (\prob(Y \leq z \mid \LL(F(z))) - \one \{ Y \leq z \})^2 \diff z,
\end{align*}
which implies $\MCB_\ISO \geq \MCB_\BS$.  Finally for any $\alpha \in (0,1)$, we have that $P_{Y \mid \LL(F)}^{-1}(\alpha)$ minimizes $\E \, \qs(X,Y)$ over all $\LL(F)$-measurable random variables $X$.  We use that $\LL(F^{-1}(\alpha)) \subseteq \LL(F)$, to derive that
\begin{align*}
\E \, \crps(P_{Y \mid \LL(F)},Y) = 
\int_0^1 \E \, \qs (P_{Y\mid \LL(F)}^{-1}(\alpha),Y) \diff \alpha \leq \int_0^1 \E \, \qs (q_\alpha (Y \mid \LL(F^{-1}(\alpha)),Y) \diff \alpha
\end{align*}
and hence $\MCB_\ISO \geq \MCB_\QS$.
\end{proof}

\begin{proof}[Proof of Proposition \ref{prop:properties_Hersbach_decomp}]  
The claim in part (a) follows from the definition of MS at \eqref{eq:MS_Hersbach}.  For part (b), suppose that $F$ is auto-calibrated.  Then $Y \in \supp(F)$ almost surely and hence $\MS = 0$ by part (a).  The tower property implies for any $A \in \BU$ that 
\begin{align*}
\tau(A) & = \E \left( \E \left( \int_A \one \{ F(Y) \leq p \} \diff \nu_F(p) \MMid F \right) \right) \\
        & = \E \left( \int_A \E \left( \one \{ F(Y) \leq p \} \mid F \right) \diff \nu_F(p) \right) \\
        & = \E \left( \int_A F(F^{-1}(p)) \diff \nu_F(p) \right),
\end{align*}
where the last equality follows since if $Y \in \supp(F)$, then $F(Y) \leq p$ if and only if $Y \leq F^{-1}(p)$ and $\prob(Y \leq F^{-1}(p) \mid F) = F(F^{-1}(p))$ by auto-calibration.  By the properties of generalized inverses \citep{Generalized_inverses}, we have $F(F^{-1}(p)) \geq p$ for all $p\in (0,1)$.  However, if $F(F^{-1}(p)) > p$ for all $p\in B$ in some $B \in \BU$, then $F^{-1}(B) = \{ x \in \R \mid F(x) \in B \} = \emptyset$ and hence $\nu_F(B) = 0$ almost surely.  That is, $\nu_F( \{ p \in (0,1) : F(F^{-1}(p)) > p \} = 0$ almost surely and thus 
\begin{align*}
\tau(A) = \E \left( \int_A F(F^{-1}(p)) \diff \nu_F(p) \right) = \E \left( \int_A p \diff \nu_F(p) \right) = \int_A p \diff \mu(p).
\end{align*}
We conclude that $f(p) = p$ $\mu$-almost surely and hence $\MCB_\HB = 0$.

The condition in part (c) is equivalent to assuming that $\frac{\diff}{\diff p} F^{-1}$ is almost surely constant for all $p \in (0,1)$.  Since $F$ is probabilistically calibrated, we have for any $p\in (0,1)$,
\begin{align*}
f(p) 
= \frac{1}{\gamma(p)} \E \left( \one \{ F(Y) \leq p \} \frac{\diff}{\diff p} F^{-1}(p) \right) 
= \frac{\gamma (p)}{\gamma(p)} \E \left( \one \{ F(Y) \leq p \} \right) = \prob (F(Y) \leq p) = p
\end{align*}
and hence $\MCB_\HB = 0$. 
\end{proof}

\section{Analytic examples at the population level}  \label{app:examples}

In this section we compare the population level decompositions from Section \ref{sec:population} in a number of examples in the prediction space setting.  Table \ref{tab:examples} collects and summarizes the analytic forms of the decomposition components in these examples.  Assumption \ref{assump:population} is satisfied throughout.

\renewcommand{\arraystretch}{1.3}
\begin{table}[t] 
\centering
\caption{Analytic form of the various different types of decomposition in population level examples E.1, \dots, E.5.  For details and supporting calculations see the text.  \label{tab:examples}}
\medskip
\begin{tabular}{lcccccc} 
\toprule
Example & E.1 & E.2 & E.3 & E.4 & E.5 \\
\midrule
$\E \, \crps(F,Y)$ & $\sum_{i=1}^n w_{i} \frac{\sigma_{i}}{\sqrt{\pi}}$ & $\frac{1}{6}$ & 1 & $\frac{39}{80}$ & $\frac{5}{24} t$ \\
\midrule
$\UNC_0$ & $\frac{1}{2} \sum_{i,j = 1}^n w_i w_j \, A(\mu_i-\mu_j, \sigma_i^2 + \sigma_j^2)$ & $\frac{2}{5}$ & $\frac{3}{4}$ & $\frac{3}{2}$ &  $\frac{2}{9} t$ \\
\midrule
$\MCB_\CT$  & 0 & $\frac{1}{30}$ & 1               & $\frac{7}{400}$  & $\frac{3}{200} t$ \\

$\MCB_\ISO$ & 0 & $\frac{1}{30}$ & 1               & $\frac{9}{2800}$ & $\frac{3}{200} t_2$ \\
$\MCB_\QS$  & 0 & $\frac{1}{30}$ & $\frac{13}{16}$ & $\frac{9}{2800}$ & 0 \\
$\MCB_\BS$  & 0 & $\frac{1}{30}$ & $\frac{1}{2}$   & $\frac{9}{2800}$ & 0 \\
$\MCB_\HB$  & 0 & 0              & $\frac{1}{8}$   & $\frac{1}{1600}$ & 0 \\
\bottomrule
\end{tabular}
\end{table}

\subsection{Auto-calibrated Gaussian}  \label{app:1}

In this example, the predictive distribution $F$ is Gaussian with mean $\mu_i$ and standard deviation $\sigma_i > 0$ with probability $w_i$ for $i = 1, \dots, n$, where $w_i + \dots + w_n = 1$.  Conditionally on $F$, the outcome $Y$ has distribution $F$, so $F$ is auto-calibrated.  We conclude that 
\begin{align*}
\MCB_\CT = \MCB_\ISO = \MCB_\BS = \MCB_\QS = 0.  
\end{align*}
Since auto-calibration implies probabilistic calibration, Proposition \ref{prop:properties_Hersbach_decomp} yields $\MCB_\HB = \MS_\HB = 0$.  Finally, we apply formulas in \cite{Grimit_et_al} to obtain 
\begin{align*}
\E \, \crps(F,Y) = \sum_{i=1}^n w_i \frac{\sigma_i}{\sqrt{\pi}}
\quad \textrm{and} \quad 
\UNC_0 = \frac{1}{2} \sum_{i,j=1}^n w_i w_j A(\mu_i-\mu_j,\sigma_i^2+\sigma_j^2), 
\end{align*}
where $A(\mu,\sigma^2) = 2 \sigma \phi(\frac{\mu}{\sigma}) + \mu (2 \Phi (\frac{\mu}{\sigma}) - 1)$, with $\phi$ and $\Phi$ denoting the density and the \cdf\ of the standard normal distribution, respectively. 

\subsection{Example in \cite{Candille_2005}}  \label{app:2}

In this example of \citet[p.~2145]{Candille_2005}, the forecast $F$ is $F_1$, which is uniform on $(-1,0)$, or $F_2$, which is uniform on $(0,1)$, with equal probability.  Given $F = F_1$, the conditional \cdf\ of $Y$ is $Q_1(z) = 1-z^2$ for $z \in (-1,0)$, and given $F = F_2$, the conditional \cdf\ of $Y$ is $Q_2(z) = z^2$ for $z \in (0,1)$. 

For $i = 1, 2$, we denote by $G_i$ the isotonic conditional law of $Y$ given $F = F_i$.  Since $F_1 \st F_2$ and $Q_1 \st Q_2$ it follows that $Q_i = G_i$ for $i = 1, 2$ and the isotonicity-based decomposition coincides with the Candille--Talagrand decomposition.  For any $z \in (-1,1)$, $F_1(z)$ and $F_2(z)$ strictly order and hence the random variable $F(z)$ already reveals the value of $F$.  That is, $\sigma(F(z)) = \sigma(F)$ and hence $\prob(Y \leq z \mid F(z)) = \prob(Y \leq z \mid F) = P_{Y \mid F}(z)$.  Since this conditional probability is already an increasing function of $F(z)$, we may conclude by Proposition 3.2. in \citet{ICL} that $\prob(Y \leq z \mid \LL(F(z))) = P_{Y \mid F}(z)$ for all $z \in \R$ and hence the Brier score based decomposition correspond with the Candille--Talagrand decomposition.  Analogously the claim can be shown for the quantile score based decomposition.  Thus the isotonicity-based, Brier score based, and quantile score based decompositions coincide with the Candille--Talagrand decomposition, where $\E \, \crps(F,Y) = 1/6$, $\MCB_\CT = 1/30$, and $\UNC_0 = 2/5$.

The forecasts satisfy the conditions in part (c) of Proposition \ref{prop:properties_Hersbach_decomp}, therefore $\MCB_\HB = 0$.  Since $Y \in \supp(F)$ almost surely, we have $\MS = 0$.

\subsection{Example with two atoms}  \label{app:3}

This simple example illustrates that the Brier score and quantile score based decompositions do not coincide in general, that the corresponding calibration methods do not necessarily produce valid {\cdf}s or quantile functions, respectively, and that $\DSC_\HB$ can be negative.

Consider the distributions $F_1 = (\delta_1 + \delta_2)/2$ and $F_2 = (\delta_0 + \delta_3)/2$, where $\delta_z$ denotes the Dirac measure at $z\in \R$.  Assume that $F$ is $F_1$ and $F_2$ with equal probability and that $Y = y_1$ if $F = F_1$ and $Y = y_2$ if $F = F_2$.  Let $y_1 = 3$ and $y_2 = 0$, so the marginal law $F_\mg$ of $Y$ is $F_2$.  We readily compute $\E \, \crps(F,Y) = 1$ and $\E \, \crps(F_\mg,Y) = \UNC_0 = 3/4$.  

An application of the PAV algorithm for the mean functional on $(\one \{ y_1 \leq z \}, \one \{ y_2 \leq z\})$ with respect to the order induced by $(F_1(z), F_2(z))$ at threshold $z \in \R$ results in  
\begin{align*}  \textstyle
\cF_1(z) = \frac{1}{2} \one_{[1,3)}(z) + \one_{[3,\infty)}(z)
\eqand
\cF_2(z) = \one_{[0,1)}(z) + \frac{1}{2} \one_{[1,3)}(z) + \one_{[3,\infty)}(z),
\end{align*}
and we see that $\cF_2$ fails to be increasing.  Similarly, an application of the PAV algorithm for the $\alpha$-quantile on $(y_1,y_2)$ with respect to the order induced by $(F_1^{-1}(\alpha), F_2^{-1}(\alpha))$ at level $\alpha \in (0,1)$ results in 
\begin{align*}  \textstyle
\cFinv^{-1}_1(\alpha) = 3
\eqand
\cFinv^{-1}_2(\alpha) = 3 \one_{(\frac{1}{2},1]}(\alpha), 
\end{align*}
so $\cFinv^{-1}_2$ fails to be increasing.  Furthermore, it follows easily that $\MCB_\BS = 1/2 \neq 13/16 = \MCB_\QS$.  As the conditional law of $Y$ given $F$ is a Dirac measure, $\E \, \crps(P_{Y \mid F},Y) = 0$ and $\MCB_\CT = 1$.  Similarly, $\MCB_\ISO = 1$ since $F_1$ and $F_2$ do not order.

According to the formulas in Section \ref{sec:HB}, $\bar{g}_1 = 2$ and $\bar{f}_1 = ( \one \{F_1(y_1) \leq \frac{1}{2}\} + 3 \one \{ F_2(y_2) \leq 1/2 \}) / (2 \bar{g}_1) = 3/4$ and thus $\MCB_{\HB} = (p_1 - \bar{f}_1)^2 \, \bar{g}_1 = 1/8$, whence we conclude that $\DSC_{\HB} = \MCB_{\HB} + \UNC_0 - \E \, \crps(F,Y) = - 1/8$. 

\subsection{Example 2.4 a) in \citet{Tilmann_Johannes_Calibration}}  \label{app:4}

Let $F$ be a mixture of uniform distributions on $[0, 1]$, $[1, 2]$, and $[2, 3]$ with weights $p_1, p_2$, and $p_3$, respectively, and let $Y$ be drawn from a mixture of these distributions with weights $q_1, q_2$, and $q_3$, respectively, where the tuple $(p_1, p_2, p_3;q_1,q_2,q_3)$ attains each of the values
\begin{align*}  \textstyle
\left( \frac{1}{2},\frac{1}{4},\frac{1}{4}; \frac{5}{10},\frac{1}{10},\frac{4}{10}\right), \quad 
\left( \frac{1}{4},\frac{1}{2},\frac{1}{4}; \frac{1}{10},\frac{8}{10},\frac{1}{10}\right), \quad 
\left( \frac{1}{4},\frac{1}{4},\frac{1}{2}; \frac{4}{10},\frac{1}{10},\frac{5}{10}\right)
\end{align*}
with equal probability.  We note that $F$ is probabilistically calibrated, and still we find that $\MCB_\HB \not= 0$. 

Let $F_1, F_2$, and $F_3$ denote the distributions that $F$ attains.  For $i = 1, 2, 3$, let $Q_i$ be the conditional law of $Y$ given $F = F_i$, and let $G_i$ be the isotonic conditional law of $Y$ given $F = F_i$.  The marginal law $F_\mg$ of $Y$ is uniform on $[0,3]$ and, hence, 
\begin{align*}  \textstyle
\UNC_0 = \E \, \crps (F_\mg,Y) 
& = \int \int (F_\mg(x) - \one \{ y \leq x \})^2 \diff x \diff F_\mg(y) \\
& = \frac{1}{3} \int_0^3 \int_0^3 \left( \frac{x}{3} - \one \{ y \leq x \} \right)^2 \diff x \diff y = \frac{3}{2}.
\end{align*}
It holds that $F_1 \st F_2 \st F_3$ but only $Q_1 \st Q_3$, hence $P_{Y \mid F} \not= P_{Y \mid \LL(F)}$.  Let $r = 10/7$, $s = 11/7$.  On $(-\infty,r]$, we have the pointwise inequalities $Q_2\le Q_3 \le Q_1$; on $[r,s]$, we have $Q_3 \le Q_2 \le Q_1$; and on $[s,\infty)$, we have $Q_3 \le Q_1 \le Q_2$.  Consider the pooled {\cdf}s $Q_{12} = (Q_1 + Q_2)/2$ and $Q_{23} = (Q_2 + Q_3)/2$.  The $G_i$'s may be derived by pooling the $Q_i$'s according to the given order constraint $G_1 \st G_2 \st G_3$, namely,
\begin{align*}  \textstyle
& G_1(z) = Q_1(z) \one_{(-\infty,s]}(z) + Q_{12}(z) \one_{[s,\infty)}(z), \\
& G_2(z) = Q_{23}(z) \one_{(-\infty,r]}(z) + Q_2(z) \one_{[r,s]}(z) + Q_{12}(z) \one_{[s,\infty)}(z), \\
& G_3(x) = Q_{23}(z) \one_{(-\infty,r]}(x) + Q_{3}(z) \one_{[r,\infty)}(z).
\end{align*}
By the law of total expectation and Fubini's theorem,
\begin{align*}  \textstyle
\E \, \crps(F,Y) 
& = \frac{1}{3} \sum_{i=1}^3 \E \, \big(\crps(F,Y) \mid F=F_i\big) \\
& = \frac{1}{3} \sum_{i=1}^3 \int \int \big( F_i(x) - \one \{ y \leq x \} \big)^2 \diff x \diff Q_i(y) \\
& = \frac{1}{3} \sum_{i=1}^3 \int \int \big( F_i(x) - \one \{ y \leq x \} \big)^2 \diff Q_i(y) \diff x \\
& = \frac{1}{3} \sum_{i=1}^3 \int  \left( F_i^2(x) - 2F_i(x)Q_i(x) + Q_i(x) \right) \diff x.
\end{align*}
Similarly, we find that $\E \, \crps(G,Y) = (1/3) \sum_{i=1}^3 \int (G_i^2(x) - 2G_i(x)Q_i(x) + Q_i(x)) \diff x$ and $\E \, \crps(Q,Y) = (1/3) \sum_{i=1}^3 \int (Q_i(x) - Q_i^2(x)) \diff x$; hence $\E \,  \crps(F,Y) = 39/80$, $\E \, \crps(G,Y)$ $= 339/700$, and $\E \, \crps(Q,Y) = 47/100$.  We conclude that
\begin{align*}  \textstyle
\MCB_\CT = \frac{39}{80} - \frac{47}{100} = \frac{7}{400} \eqand \MCB_\ISO = \frac{39}{80} - \frac{339}{700} = \frac{9}{2800}.
\end{align*}
Since the predictive distributions are ordered with respect to $\st$, it follows that for every threshold $z$, the ordering of $F_i(z)$ is the same.  For $z \in (-\infty,1]$, $F_2(z)$ and $F_3(z)$ coincide but this also holds for $G_2(z)$ and $G_3(z)$. Similarly, for $z \in [2,\infty)$, $F_1(z)$ and $F_2(z)$ coincide but this also holds for $G_1(z)$ and $G_2(z)$.  This implies that the Brier score based and the isotonocity-based decompositions coincide.  Since the stochastic order is equivalently characterized by pointwise orderings of lower quantile functions, the quantile score based and the isotonicity-based decompositions also coincide. 

As all $F_i^{-1}$'s are absolutely continuous, we may apply Corollary \ref{cor:HB_1} to compute $\MCB_\HB$.  For $p \in (0,1) \setminus \{ 1/4, 1/2, 3/4\}$ we find that
\begin{align*}  \textstyle
& \frac{\diff}{\diff p} F_1^{-1}(p) = 2 \one_{(0,\frac{1}{2})}(p) + 4 \one_{(\frac{1}{2},1)}(p), \quad \frac{\diff}{\diff p} F_3^{-1}(p) = 4 \one_{(0,\frac{1}{2})}(p) + 2 \one_{(\frac{1}{2},1)}(p), \\
& \frac{\diff}{\diff p} F_2^{-1}(p) = 4 \one_{(0,\frac{1}{4})}(p) + 2 \one_{(\frac{1}{4},\frac{3}{4})}(p) + 4 \one_{(\frac{3}{4},1)}(p),
\end{align*}
hence 
\begin{align*}  \textstyle
\gamma(p) 
= \frac{1}{3} \sum_{i=1}^3 \E \left( \frac{\diff}{\diff p} F^{-1}(p) \Big| F = F_i \right)  
= \frac{10}{3} \one_{(0,\frac{1}{4})}(p) + \frac{8}{3} \one_{(\frac{1}{4},\frac{3}{4})}(p) + \frac{10}{3} \one_{(\frac{3}{4},1)}(p).
\end{align*}
The law of total expectation implies 
\begin{align*}  \textstyle
\E \left( \one \{ F(Y) \leq p \} \frac{\diff}{\diff p} F^{-1}(p) \right) 
& = \frac{10}{3} p \one_{(0,\frac{1}{4})}(p) + \left( \frac{3}{15} + \frac{34}{15}p\right) \one_{(\frac{1}{4},\frac{3}{4})}(p)+ \frac{10}{3} p\one_{(\frac{3}{4},1)}(p),
\end{align*}
and hence,
\begin{align*}  \textstyle
f(p) = p \one_{(0,\frac{1}{4})}(p) + \left( \frac{3}{40} + \frac{17}{20}p \right) \one_{(\frac{1}{4},\frac{3}{4})}(p) + p \one_{(\frac{3}{4},1)}(p).
\end{align*}
Finally, we obtain 
\begin{align*}  \textstyle
\MCB_\HB = \int \left( p - f(p) \right)^2 \gamma (p) \diff p  = \int_{\frac{1}{4}}^{\frac{3}{4}} \left(\frac{3}{20}p - \frac{3}{40}\right)^2 \frac{8}{3} \diff p = \frac{1}{1600}.
\end{align*}
  
\renewcommand{\thesubsubsection}{\thesection.\arabic{subsection}.\alph{subsubsection}}

\subsection{Example 2.14 b) in \cite{Tilmann_Johannes_Calibration}}  \label{app:5}

For $y_1 < y_2 < y_3$, let $F$ be a mixture of the Dirac measures on $y_1, y_2$, and $y_3$ with weights $p_1, p_2$, and $p_3$, and let $Y$ be drawn from a mixture of the same Dirac measures with weights $q_1, q_2$, and $q_3$, respectively.  Suppose that the tuple $(p_1, p_2, p_3; q_1, q_2, q_3)$ attains each of the values
\begin{align*}  \textstyle
\left( \frac{1}{2},\frac{1}{4},\frac{1}{4}; \frac{5}{10},\frac{4}{10},\frac{1}{10} \right), \quad 
\left( \frac{1}{4},\frac{1}{2},\frac{1}{4}; \frac{1}{10},\frac{5}{10},\frac{4}{10} \right), \quad 
\left( \frac{1}{4},\frac{1}{4},\frac{1}{2}; \frac{4}{10},\frac{1}{10},\frac{5}{10} \right)
\end{align*}
with equal probability.  Let $t_1 = y_2 - y_1 > 0$, $t_2 = y_3 - y_2 > 0$, and $t = t_1 + t_2$.  It is immediate that $\E \, \crps(F,Y) = 5t/24$ and $\UNC_0 = \E \, \crps (F_\mg,Y) = 2t/9$.  As \citet{Tilmann_Johannes_Calibration} show, $F$ is threshold and quantile calibrated, hence $\MCB_\BS = \MCB_\QS = 0$. 

Let $F_1, F_2$, and $F_3$ denote the three discrete distributions that $F$ may attain.  For $i = 1, 2, 3$, denote by $Q_i$ the conditional law of $Y$ given $F = F_i$ and by $G_i$ the isotonic conditional law of $Y$ given $F = F_i$, namely, 
\begin{align*}  \textstyle
G_1 = \frac{1}{2} \delta_{y_1} + \frac{4}{10} \delta_{y_2} + \frac{1}{10} \delta_{y_3}, \quad
G_2 = \frac{1}{4} \delta_{y_1} + \frac{7}{20} \delta_{y_2} + \frac{4}{10} \delta_{y_3}, \quad 
G_3 = \frac{1}{4} \delta_{y_1} + \frac{1}{4} \delta_{y_2} + \frac{1}{2} \delta_{y_3}.
\end{align*}
Since the image of the random vector $(F,Y)$ is finite and ICL is the population version of IDR \citep[Proposition 4.1]{ICL}, one obtains the $G_i$'s alternatively by applying IDR on the finite sample of size $n = 30$ with five occurrences of $(F_1,y_1)$, four of $(F_1,y_2)$, one each of $(F_1, y_3$ and $(F_2,y_1)$, five of $(F_2,y_2)$, four each of $(F_2,y_3)$ and $(F_3,y_1)$, one of $(F_3,y_2)$, and five of $(F_3,y_3)$.  The $\MCB_\CT$ and $\MCB_\ISO$ components may be calculated in analogy to previous examples.  We obtain $\MCB_\CT = 3t/200$ and $\MCB_\ISO = 3t_2/200$. 

To compute the Hersbach decomposition, let $\nu_i$ be the image of the Lebesgue measure on $(0,1)$ under $F_i$ where $i = 1, 2, 3$.  We have $\nu_1 = t_1 \delta_{1/2} + t_2 \delta_{3/4}$, $\nu_2 = t_1 \delta_{1/4} + t_2 \delta_{3/4}$, and $\nu_3 = t_1 \delta_{1/4} + t_2 \delta_{1/2}$, and hence, $\mu = (1/3) (2 t_1 \, \delta_{1/4} + t \, \delta_{1/2} + 2t_2 \, \delta_{3/4})$.  For $\ell = 1, 2, 3$ and $p_l = l/4$, and for any $A \in \BU$, the quantities $f_\ell = f(p_\ell)$ satisfy 
\begin{align}  \textstyle
\tau (A) 
& = \E \int_A \one \{ F(Y) \leq p \} \diff \nu_F(p) \label{eq:app:5b} \\
& = \int_A f(p) \diff \mu(p) = f_1 \, \frac{2 t_1}{3} \, \delta_{1/4}(A) + f_2 \, \frac{t}{3} \, \delta_{1/2}(A) + f_3 \, \frac{2t_2}{3} \, \delta_{3/4}(A) \nonumber,
\end{align}
where the expectation in \eqref{eq:app:5b} may be calculated by the law of total expectation:
\begin{align*}  \textstyle
\E \int_A \one \{ F(Y) \leq p \} \diff \nu(p)
& = \frac{1}{3} \sum_{i=1}^3 \E \left( \int_A \one \{ F(Y) \le p \} \diff \nu_F(p) \Mid F = F_i \right) \\
& = \frac{1}{3} \sum_{i=1}^3 \int \int_A \one \{ F_i(y) \le p \} \diff \nu_i(p) \diff Q_i(y) \\
& = \frac{t_1}{6} \, \delta_{1/4}(A) + \frac{t}{6} \, \delta_{1/2}(A) + \frac{t_2}{2} \, \delta_{3/4}(A).
\end{align*}
We conclude that $f_\ell = p_\ell$ for $\ell = 1, 2, 3$, and hence $\MCB_\HB = 0$.

\end{document}